\documentclass[twocolumn,aps,prl,reprint,superscriptaddress,amsmath,amssymb,bbm,floatfix]{revtex4-2}
\usepackage[latin9]{inputenc}
\usepackage{color}
\usepackage{amstext}
\usepackage{graphicx}
\usepackage[colorlinks, allcolors={blue}]{hyperref}
\usepackage{physics}
\usepackage{blindtext}

\usepackage{xcolor,graphicx}
\usepackage{newlfont}
\usepackage{amssymb,amsmath,mathrsfs,amsthm}
\usepackage{verbatim}
\usepackage{bbm}
\usepackage{multirow}
\usepackage[varg]{txfonts}

\newcommand{\orcid}[1]{\href{https://orcid.org/#1}{\includegraphics[width=7pt]{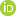}}}
\newcommand{\mc}{\mathcal}
\newcommand{\be}{\begin{equation}}
\newcommand{\ee}{\end{equation}}
\newcommand{\bbm}{\mathbbm}
\newcommand{\I}{\mathbbm{1}}
\newcommand{\ic}{\mathbbm{i}}
\newcommand{\mb}{\mathbb}

\newtheorem{theorem}{Theorem}

\newtheorem{fact}{Fact}
\newtheorem{conjecture}{Conjecture}

\renewcommand\bra[1]{{\langle{#1}|}}
\makeatletter
\renewcommand\ket[1]{%
  \@ifnextchar\bra{\k@t{#1}\!}{\k@t{#1}}%
}
\newcommand\k@t[1]{{|{#1}\rangle}}
\newcommand\proj[1]{{|{#1}\rangle\!\langle{#1}|}}
\makeatother

\begin{document}

\title{Certifying classes of $d$-outcome measurements with quantum steering}

\author{Alexandre C. Orthey Jr. \orcid{0000-0001-8111-3944}}
\affiliation{Center for Theoretical Physics, Polish Academy of Sciences, Al. Lotnik\'ow 32/46, 02-668 Warsaw, Poland.}
\affiliation{Institute of Fundamental Technological Research, Polish Academy of Sciences, Pawi\'nskiego 5B, 02-106 Warsaw, Poland.}

\author{Remigiusz Augusiak \orcid{0000-0003-1154-6132}}
\email{augusiak@cft.edu.pl}
\affiliation{Center for Theoretical Physics, Polish Academy of Sciences, Al. Lotnik\'ow 32/46, 02-668 Warsaw, Poland.}

\begin{abstract}

Device-independent certification schemes are based on minimal assumptions about the quantum system under study, which makes the most desirable among certification schemes. However, they are often the most challenging to implement. In order to reduce the implementation cost one can consider semi-device-independent schemes such as those based on quantum steering. Here we provide a construction of a family of steering inequalities which are tailored to large classes of $d$-outcomes projective measurements being a certain linear combination of the Heisenberg-Weyl operators on the untrusted side and a fixed set of known measurements on the trusted side. We then prove that the maximal quantum violation of those inequalities can be used for certification of those measurements and the maximally entangled state of two qudits. Importantly, in our self-testing proof, we do not assume the shared state to be pure, nor do we assume the measurements to be projective. Before concluding, we also show how robust to noise our self-testing statement is. We believe that our construction broadens the scope of semi-device-independent certification, paving the way for more general but still less costly quantum certification protocols.
%
\end{abstract}

\maketitle

\section{Introduction}

The advantage of quantum strategies over classical ones when performing certain tasks can be verified by violating Bell-like inequalities \cite{brunner2014bell}. These inequalities are tailored according to the scenario under scrutiny in such a way that overcoming the bound achieved by any classical strategy necessarily implies that a certain quantum correlation must be present in the system. By maximally violating these inequalities it is possible to go beyond that and certify more specific details about the quantum system, such as measurements and states, up to local unitary transformations and extra degrees of freedom. In fact, that can always be done by simply performing quantum tomography on the system, but such a procedure often requires too many measurements and resources and, more importantly, requires the measuring devices to be fully characterized and to perform known measurements. 
With Bell-like inequalities, on the other side, we can certify measurements and states in a device-independent (DI) way, that is, by not relying on the specific details of the measurement apparatus. Rather, in DI certification, we only assume classical inputs and outputs in the experiment (i.e. probabilities), which is a great advantage in the certification of quantum devices in quantum cryptography schemes \cite{Acin2007device}.

The DI certification scheme based on the maximal violation of Bell-like inequalities described above is known as self-testing \cite{mayers1998quantum,mayers2004self} (see Ref. \cite{supic2020self} for a review). Since the last decade, many inequalities were proposed to self-test any pure bipartite and multipartite entangled states, as well as, some specific classes of quantum measurements
\cite{mckague2012robust,Reichardt2013classical,McKague2014self,Wu2014robust,Bamps2015sos,Wang2016all,Supic2016self-testing,Coladangelo2017all,kaniewski2019maximal,Mancinska2021constant,Tavakoli2021mutually,Sarkar2021npj,Supic2023quantum,Das2022robust}. Among these classes of measurements, we can highlight the mutually unbiased measurements \cite{Tavakoli2021mutually} and the optimal Collins-Gisin-Linden-Massar-Popescu (CGLMP) measurements \cite{Sarkar2021npj}. In fact, a general method to certify any real local rank-one projective measurement was just recently explored \cite{chen2023all} (see also \cite{sarkar2023universal} for a general method to certify any quantum state, projective measurement and rank-one extremal non-projective measurements).

Although very general, the DI self-testing schemes of Ref. \cite{sarkar2023universal} and \cite{chen2023all} are at the same time very expensive to implement, e.g. in terms of the number of measurements that parties must perform to certify a given state or measurement. For instance, 
the scheme of Ref. \cite{chen2023all} requires that both parties perform altogether $O(d^3)$ measurements to certify a given real projective measurement acting on a Hilbert space of dimension $d$. 
%
%
One way to reduce the complexity of such certification schemes is by making assumptions about the devices
used for certification, which characterize the so-called semi-device-independent (SDI) certification schemes. One example of these schemes is the prepare and measure (PM) scenario, where we assume the dimension of the system to be known. The PM scenario allows the certification, for instance, of nonprojective qubit measurements \cite{Tavakoli2021mutually,Mironowicz2019experimentally}, mutually unbiased basis (MUBs) \cite{Farkas2019mub}, and nonprojective symmetric informationally complete (SIC) measurements \cite{Tavakoli2019enabling}.

Another way to perform SDI certification is by making use of quantum steering \cite{wiseman2007steering,cavalcanti2009experimental,costa2016quantification,uola2019quantum} rather than Bell-nonlocality. In a steering scenario, we trust the measurements performed by one of the parties, say Alice, while assuming that the other party, say Bob, performs arbitrary unknown measurements. The use of steering for self-testing was first proposed in Ref. \cite{supic2016self_EPR}, with further developments including the certification of some classes of genuinely multipartite entangled (GME) states or incompatible measurements in $d$-dimensional space \cite{sarkar2022certification}. In the latter case, however, for any different set of incompatible measurements 
to be certified, the trusted party must also perform different measurements.

Inspired by Refs. \cite{sarkar2022certification,kaniewski2019maximal}, our aim here is to
use steering to construct certification methods for a large class of quantum measurements in $d$-dimensional space given by linear combinations of the Heisenberg-Weyl (HW) operators. Here, however, we consider a scenario in which, unlike in Ref. \cite{sarkar2022certification}, Alice always performs the same fixed measurements (see Fig. \ref{fig:scenario}) independently of the measurements to be certified on Bob's side.
%
The state giving rise to the maximal quantum value is always the maximally entangled state of two qudits. 
\begin{figure}[h]
    \centering
    \includegraphics[width=0.8\linewidth]{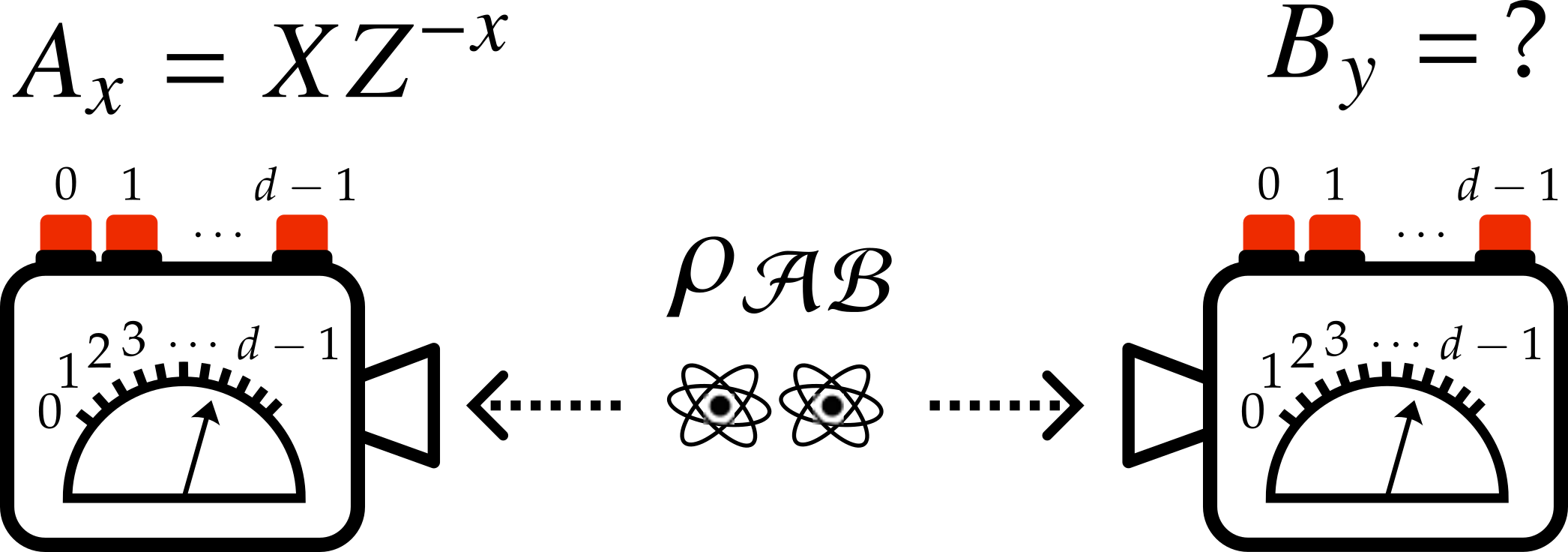}
    \caption{Depiction of the scenario we consider: both Alice and Bob perform $d$ measurements with $d$-outcome each on an arbitrary bipartite state $\rho_\mc{AB}$. We assume that Alice's share of the state is of known local dimension $\dim \mc{H_A}=d$. Alice's measurements are projective and given by $A_x=XZ^{-x}$, while Bob's measurements $B_y$ are arbitrary. Our goal is to certify that the state is equivalent to the maximally entangled state of local dimension $d$ and that Bob's measurements are equivalent to any appropriate unitary linear combination of $XZ^k$, with $k\in\{0,\ldots,d-1\}$.}
    \label{fig:scenario}
\end{figure}
%
%
%
With that in mind, we derive steering inequalities tailored to every different linear
combination. Our steering inequalities are maximally violated by the maximally entangled state of local dimension $d$ and measurements that are equivalent to the given linear combinations of HW operators. 
%
%
It is important to note that our self-testing statement assumes that the state shared by the parties can be mixed and that Bob's measurements can be non-projective.

Our paper is organized as follows. In Sec. \ref{sec_pre} we provide a review on generalized observables, Heisenberg-Weyl operators, and the definition of self-testing. In Sec \ref{sec_results} we present our steering inequality, provide an example for qutrits, calculate the quantum bound, present the self-testing statement, and end with the calculation of the classical bound for qutrits. The robustness analysis is presented in Sec. \ref{sec_robustness}. Our conclusions are left for Sec. \ref{sec_discussion}

\section{Preliminaries}\label{sec_pre}

\subsection{Generalized observables}

The scenario in question is composed of a simple bipartite system where the two parties, Alice and Bob, represented by curly symbols $\mc{A}$ and $\mc{B}$, share a bipartite quantum state $\rho_\mc{AB}$ acting over $\mc{H_A\otimes H_B}$. Both of them perform $d$ measurements of $d$ outcomes each. We assume that Alice's measurements are trusted, that is, completely characterized. Bob's measurements, however, are unknown; we assume that they can be non-projective and of any dimension. We also assume that the state shared by Alice and Bob can be mixed, but Alice's share is of local dimension $\dim \mc{H_A}=d$.

Alice's and Bob's measurements are represented by POVM's $M_x=\{M_{a|x}\}_{a=0}^{d-1}$ and $N_y=\{N_{b|y}\}_{b=0}^{d-1}$, that is, $M_{a|x}\geqslant 0$ and $\sum_{a=0}^{d-1} M_{a|x}=\I_d$, and similarly for Bob. The measurement choices and the outcomes are identified by $x,y\in\{0,\ldots,d-1\}$ and $a,b\in\{0,\ldots,d-1\}$, respectively. After performing the experiment many times, the statistics can be collected in the form of a set of probabilities $\mathbf{p}\coloneqq\{p(ab|xy)\}_{abxy}$, henceforth called \textit{behavior}, that are calculated using Born rule, i.e., $p(ab|xy)=\Tr[(M_{a|x}\otimes N_{b|y})\rho_\mc{AB}]$.

To detect steering, we are going to use Bell-type inequalities, that is,
\be
\sum_{a,b,x,y}\alpha_{a,b,x,y}p(a,b|x,y)\leqslant \beta_C,
\ee
where $\alpha_{a,b,x,y}$ are coefficients suitably chosen and $\beta_C$ is the bound satisfied by behaviors that admit a local hidden state model \cite{uola2019quantum}; we will refer to it simply as \textit{classical bound}. A behavior that violates the above steering inequality implies that the state $\rho_\mc{AB}$ is \textit{steerable}.

Since our goal is to deal with $d$-dimensional systems, it is convenient to work with generalized correlators, which are the Fourier transform of $\mathbf{p}$ given by
\be\label{gen_correlators}
\langle A_{k|x}B_{l|y} \rangle=\sum_{a,b=0}^{d-1} \omega^{ak+bl}p(ab|xy),
\ee
where $\omega=\exp(2\pi\ic/d)$ and $k,l\in\{0,\ldots,d-1\}$. The objects $A_{k|x}$ and $B_{l|y}$ are called \textit{generalized observables}; they are represented as the Fourier transforms of the original measurement operators, that is, $A_{k|x}=\sum_{a=0}^{d-1} \omega^{ak} M_{a|x}$ and $N_{l|y}=\sum_{b=0}^{d-1} \omega^{bl} N_{b|y}$. Because the Fourier transform is invertible, $A_x=\{A_{k|x}\}$ and $B_y=\{B_{l|y}\}$ fully characterize Alice's and Bob's measurements. Therefore, we can make all the calculations and proofs with $A_x$ and $B_y$, and at the end of the day, we do the inverse Fourier transform to send $M_x$ and $N_y$ to our experimentalist colleagues so they can perform the measurements in the lab. Moreover, in the quantum case we have
\be
\langle A_{k|x}B_{l|y} \rangle=\Tr\left[\rho_{AB}\left(A_{k|x}\otimes B_{l|y} \right)\right].
\ee

Because we know Alice's measurements, we can assume that they are projective without loss of generality. It is proven in Appendix A of Ref. \cite{kaniewski2019maximal} that if $M_x$ is projective, then $A_{k|x}$ are unitary operators such that $A_{k|x}=(A_{1|x})^k$. In that case, we use the notation $A_{1|x}\coloneqq A_x$, which means that $A_{k|x}=A_x^k$. From the fact that $\omega^{d-k}=\omega^{-k}=(\omega^k)^*$, we can summarize the properties satisfied by operators $A_x$ and $B_y$ by
\begin{subequations}
\begin{align}
    & A_{x}^{d-k}=A_x^{-k}=\left(A_x^k \right)^\dagger,\label{prop_Ax}\\
    & B_{d-l|y}=B_{-l|y}=\left(B_{l|y} \right)^\dagger,
\end{align}
\end{subequations}
for every $k,l\in\{0,\ldots,d-1\}$ and also $A_x^d=\I_d$ and $B_{l|y}^d=\I_\mc{B}$. While $A_x$ are unitary, Bob's measurements satisfy $B_{l|y}^\dagger B_{l|y}\leqslant \I_\mc{B}$ for $l,y\in\{0,\ldots,d-1\}$ \cite{kaniewski2019maximal}.

In our scenario, we are going to assume for convenience that Alice's generalized observables are given by the Heisenberg-Weyl operators
\be
A_x\coloneqq XZ^{-x}=(XZ^x)^*,
\ee
for $x\in\{0,\ldots,d-1\}$, where $X$ and $Z$ are the generalized Pauli operators given by
\be
        X=\sum_{i=0}^{d-1}\ket{i+ 1}\bra{i}\qquad\text{and}\qquad Z=\sum_{i=0}^{d-1}\omega^i\ket{i}\bra{i},
\ee
such that $\ket{d}\equiv\ket{0}$. The operators $A_x$ are unitary and satisfy properties \eqref{prop_Ax}. Together with the computational basis, they provide $d+1$ mutually unbiased basis in $\mathbb{C}^d$ for prime $d$ \cite{bandyopadhyay2002new}.

\subsection{Self-testing}

By observing the maximal violation of a tailored steering inequality, it is possible, in some cases, to certify the state shared by the parties and the measurements performed by the untrusted party, up to local unitary transformations. As we mentioned in the introduction, this method can be thought of as self-testing in the SDI scenario or, more specifically, a one-sided device independent scenario (1SDI) \cite{supic2016self_EPR}. In our scenario, Alice performs known projective measurements $A_x$ and Bob performs arbitrary measurements $B_y=\{B_{k|y}\}$ over an arbitrary state $\rho_\mc{AB}$ to obtain the behavior $\mathbf{p}$. Now, consider the reference measurements $\overline{B}_y=\{\overline{B}_{k|y}\}$ and the reference state $\ket{\overline{\psi}}_\mc{AB}$ which we would like to certify. It is said that the behavior $\mathbf{p}$ semi-device-independently certify the measurements $\overline{B}_y$ and the state $\ket{\overline{\psi}}_\mc{AB}$ if there exists a unitary operator $U_\mc{B}:\mc{H_B}\to\mc{H_B}$ such that
\begin{subequations}
    \begin{align}
        \left(\I_\mc{A}\otimes U_\mc{B}\right)\rho_\mc{AB} \left(\I_\mc{A}\otimes U_\mc{B}^\dagger\right) &=\proj{\overline{\psi}}_\mc{AB'}\otimes \rho_\mc{B''},\\
        U_\mc{B} B_{l|y} U_\mc{B}^\dagger &= \overline{B}_{l|y}\otimes \I_\mc{B''},
    \end{align}
\end{subequations}
where Bob's space decomposes as $\mc{H_B}=\mc{H_{B'}}\otimes \mc{H_{B''}}$ and the auxiliary state $\rho_\mc{B''}$ acts on $\mc{H_{B''}}$. 

\section{Results}\label{sec_results}

In the following sections, we specify the measurements that we want to certify on Bob's part of the system and the steering inequality tailored for them. Then we calculate the quantum bound, present the self-testing statement, and provide the classical bound for the case of qutrits.

\subsection{Steering inequality}

Let us start by defining the measurements we want to certify on Bob's side. The eigenvectors of the Heisenberg-Weyl operators $\{XZ^k\}_{k=0}^{d-1}$ together with the computational basis provide $d+1$ mutually unbiased basis in $\mathbb{C}^d$ for prime $d$ \cite{bandyopadhyay2002new}. Therefore, by taking Bob's reference measurements as linear combinations of $XZ^k$, that is, 
\be\label{Bbar_general}
\overline{B}_y=\gamma_{y,0} X+\gamma_{y,1} XZ+\gamma_{y,2}XZ^2+\ldots+\gamma_{y,d-1}XZ^{d-1},
\ee
where $\gamma_{y,k}\in\mathbb{C}$ for every $y,k\in\{0,\ldots,d-1\}$, we can cover a large set of generalized measurements in dimension $d$. To make $\overline{B}_y$ unitary and to have correct spectrum, i.e. $\overline{B}_y^d=\I_d$, Bob's reference measurements must have the form
\be\label{Bbar_main}
\overline{B}_y=\dfrac{1}{d}\sum_{k,j=0}^{d-1}e^{ \bbm{i}\phi_{j\oplus 1,y}}\omega^{-kj}XZ^k,
\ee
where $\{\phi_{j,y}\}_{j=0}^{d-1}$ is a set of any real numbers satisfying $\sum_{j=0}^{d-1}\phi_{j,y}=0$ and $\oplus$ stands for addition modulo $d$. By choosing different values of $\phi_{j,y}$ at will, we can select different measurements to certify. In what follows, we are going to construct steering inequalities tailored for each reference measurement as a function of $\{\phi_{j,y}\}$.

The steering inequality is constructed through the sum of the generalized correlators \eqref{gen_correlators}. However, to make this sum result in a real number and to include all outcomes in the inequality we must also sum over the powers of the measurement operators. To achieve that, let us artificially define the following operators:
\be\label{B(n)_main}
\overline{B}_y^{(n)}=\sum_{k=0}^{d-1}\left(\gamma_{yk}^{(n)}\right)^* \left(XZ^k\right)^n
\ee
with $n=1,\ldots,d-1$, 
where the coefficients $\gamma_{yk}^{(n)}$ have the form
\be\label{byk_main}
\gamma_{yk}^{(n)}\coloneqq \frac{1}{d}\omega^{kn(n-1)/2}\sum_{l=0}^{d-1}e^{-\bbm{i}\sum_{m=1}^{n}\phi_{l\oplus m,y}}\omega^{knl}
\ee
and for any measurement choice $y$, the angles $\phi_{j,y}$ satisfy 
\begin{equation}
    \sum_{j=0}^{d-1}\phi_{j,y}=0.
\end{equation}
Note that we constructed $\overline{B}_y^{(n)}$ using the complex conjugation of the $\gamma$ coefficients for mathematical convenience. After some algebra (see Appendix \ref{appendix_Steering_ope}), it is possible to show that
\be
\overline{B}_y^{(n)}=\overline{B}_y^n, \qquad \text{for }n=1,\ldots,d-1,
\ee
which satisfies
\be
\left(\overline{B}_y^{(n)}\right)^d=\left(\overline{B}_y^n\right)^d=\I_d, \qquad \text{for }n=1,\ldots,d-1.
\ee
Also, note that if we take $n=d$ in \eqref{B(n)_main}, we obtain $\overline{B}_y^{(d)}=d\I$. In addition, because $\overline{B}_y$ is unitary, its powers are also unitary, as well as, $\overline{B}_y^{(n)}$ for every $n=1,\ldots,d-1$.

We want to construct a steering inequality that is maximally violated by the maximally entangled state of two qudits
\be
\ket{\phi_d^+}=\frac{1}{\sqrt{d}}\sum_{i=0}^{d-1}\ket{ii}
\ee
and the measurements $A_x$ and $\overline{B}_y$. For that, we are going to use the stabilizer set of operators of the state $\ket{\phi_d^+}$, which act in the form
\be
\left(XZ^{-k}\right)^n \otimes \left(XZ^{k}\right)^n \ket{\psi}=\ket{\psi},
\ee
for $n\in\{1,\ldots,d-1\}$ and $k\in\{0,\ldots,d-1\}$. That means that $\ket{\psi}=\ket{\phi_d^+}$ is the unique state that satisfies the above system of equations. Now, we need to isolate $(XZ^k)^n$ in \eqref{B(n)_main} to construct generalized correlators. That can be done if we impose that
\be\label{b_delta_main}
\sum_{k=0}^{d-1}\left(\gamma_{y'k}^{(n)}\right)^*\gamma_{yk}^{(n)} =\delta_{yy'},\qquad\forall y,y'.
\ee
The above condition imposes that Bob's reference measurements $\{\overline{B}_{y=1},\ldots,\overline{B}_{y=d-1}\}$ are in function of the first reference measurement $\overline{B}_{y=0}$ in terms of the complete free set of numbers $\{\phi_{j,y=0}\}_{j=0}^{d-2}$. As a consequence, we also have that (see Appendix \ref{appendix_Steering_ope} for details)
\be
\sum_{y=0}^{d-1}\left(\gamma_{yk'}^{(n)}\right)^*\gamma_{yk}^{(n)} =\delta_{kk'},\qquad\forall k,k'.
\ee
If we multiply \eqref{B(n)_main} by $\gamma_{yk'}^{(n)}$ and sum over $y$ on both sides, we can use the above relation to obtain
\be
\left(XZ^k \right)^n=\sum_{y=0}^{d-1} \gamma_{yk}^{(n)}\overline{B}_y^{(n)}.
\ee
Now, we can specify the substitution rule
\be\label{substitution_rule_main}
\left(XZ^k\right)^n\longrightarrow \widetilde{B}_k^{(n)}\coloneqq\sum_{y=0}^{d-1}\gamma_{yk}^{(n)}B_{n|y},
\ee
where $B_y=\{B_{n|y}\}$, for $y=0,\ldots,d-1$ and $n=1,\ldots,d-1$, represent the $d$-outcome arbitrary measurements performed by Bob satisfying
\be\label{Bn^dagger=Bd-n_main}
B_{n|y}^\dagger=B_{d-n|y}=B_{-n|y}\qquad\text{and}\qquad B_{n|y}^\dagger B_{n|y}\leqslant \I.
\ee
From the substitution rule, we define our steering operator $\mc{S}$ as
\be\label{Steering_operator_main}
\mc{S}=\sum_{n=1}^{d-1}\sum_{k=0}^{d-1}  A_k^n\otimes \widetilde{B}_k^{(n)} .
\ee
Therefore, we propose a steering inequality as the sum of generalized correlators in the following form
\begin{eqnarray}\label{Steering_inequality_main}
&&\sum_{n=1}^{d-1}\sum_{k=0}^{d-1} \langle A_k^n\otimes\widetilde{B}_k^{(n)} \rangle \nonumber\\
&&=\sum_{n=1}^{d-1}\sum_{k,y=0}^{d-1}\gamma_{yk}^{(n)} \langle A_k^n\otimes B_{n|y} \rangle \leqslant \beta_C,
\end{eqnarray}
where $\beta_C$ is the classical bound to be determined as a function of $\{\phi_{j,y=0}\}_{j=0}^{d-2}$. Let us notice here that while it is in general a very difficult task to determine the value of $\beta_C$, below we show that for any of our steering functionals $\beta_C<\beta_Q$, meaning that the corresponding steering inequalities are all nontrivial. In fact, we show that our steering functionals enable self-testing maximally entangled states of two qudits, which would not be possible if $\beta_C=\beta_Q$.

\subsection{Example: qutrits}

To make the definitions of the previous section clearer, let us use qutrits as an example. In this case, Alice has three measurement options every round of the experiment, which are given by $A_0=X$, $A_1=XZ^{-1}$, and $A_2=XZ^{-2}$. On the other side, Bob's measurements $B_0$, $B_1$, and $B_2$ are unknown. Because the reference measurements $\overline{B}_y$ must be given by unitary operators that satisfy $\overline{B}_y^3=\I_3$, we must have
\be
\overline{B}_y=\sum_{k=0}^2 \left(\gamma_{yk}^{(1)}\right)^*XZ^{k}\equiv\begin{pmatrix}
 0 & 0 & e^{i \phi _{0,y}} \\
 e^{i \phi _{1,y}} & 0 & 0 \\
 0 & e^{-i \left(\phi _{0,y}+\phi _{1,y}\right)} & 0     
\end{pmatrix}.
\ee
Among the above operators, only those that satisfy condition \eqref{b_delta_main} can be self-tested by the maximal violation of \eqref{Steering_inequality_main}. For $y=0$, coefficients $\phi_{0,0}$ and $\phi_{1,0}$ are completely free. For $y=1,2$, however, condition \eqref{b_delta_main} imposes that $\{\phi_{0,1},\phi_{1,1},\phi_{0,2},\phi_{1,2}\}$ must be a function of $\{\phi_{0,0},\phi_{1,0}\}$. In Sec. \ref{sec_classical_bound}, we are going to provide numerical calculations of the classical bound $\beta_C$ as a function of $\{\phi_{0,0},\phi_{1,0}\}$.

\subsection{Quantum bound}

Quantum realizations can violate \eqref{Steering_inequality_main}. Here, we calculate the maximum violation by performing an SOS decomposition of $\mc{S}$. Indeed, 
\begin{align}\label{SOS_decomposition_main}
0 &\leqslant \dfrac{1}{2}\sum_{n=1}^{d-1}\sum_{k=0}^{d-1}\left[\I- A_k^n\otimes\widetilde{B}_k^{(n)}\right]^\dagger\left[\I- A_k^n\otimes\widetilde{B}_k^{(n)}\right]\nonumber\\
&\leqslant d(d-1)\I-\mc{S},
\end{align}
which implies that the maximum value achieved by $\mc{S}$ is given by $\beta_Q\coloneqq d(d-1)$ (see Appendix \ref{appendix_QB}). If Alice and Bob perform the measurements represented by
\be
A_x=XZ^{-x}\qquad\text{and}\qquad B_{n|y}=\sum_{k=0}^{d-1}\left(\gamma_{yk}^{(n)}\right)^* \left(XZ^k\right)^n,
\ee
for $x=0,\ldots,d-1$ and $y=0,\ldots,d-1$ respectively, then operator \eqref{Steering_operator_main} becomes a sum of stabilizing operators of the state $\ket{\phi_d^+}$. After taking the expected value, the quantum bound of $d(d-1)$ is achieved.

\subsection{Self-testing proof}

We start now by developing the self-testing statements regarding our steering operator $\mc{S}$ by considering first that the state $\rho_\mc{AB}=\ket{\psi}\bra{\psi}_\mc{AB}$ is pure and that the measurements $B_y$ performed by Bob are projective. Because of that, each of those measurements can be encoded in a single unitary observable $B_y\coloneqq B_{1|y}$ that satisfies $B_y^d=\I$. Consequently, they also satisfy $B_{n|y}=B_y^n$ and
\be
B_y^{d-n}=B_y^{-n}=\left(B_y^n\right)^\dagger.
\ee
From the above and from Fact \ref{fact_b} (see Appendix \ref{appendix_Steering_ope}), we can conclude that the operators $\widetilde{B}_k^{(n)}$ defined in \eqref{substitution_rule_main} must satisfy
\be
\widetilde{B}_k^{(d-n)}=\left(\widetilde{B}_k^{(n)}\right)^\dagger.
\ee
Now, consider the state and the measurements that give the maximum violation of \eqref{Steering_inequality_main}. We can infer from the SOS decomposition \eqref{SOS_decomposition_main} that such state and measurements must satisfy
\be\label{AB_tilde_main}
A_k^n\otimes \widetilde{B}_k^{(n)}\ket{\psi}_\mc{AB}=\ket{\psi}_\mc{AB}
\ee
for every $k=0,\ldots,d-1$ and $n=1,\ldots,d-1$. By implementing the notation $\widetilde{B}_k\coloneqq \widetilde{B}_k^{(1)}$, we can write
\be\label{ABtilde_main}
A_k\otimes \widetilde{B}_k\ket{\psi}_\mc{AB}=\ket{\psi}_\mc{AB}.
\ee

From the above, we can present the following theorem:
\begin{theorem}\label{theorem_1_main}
Assume that the steering inequality \eqref{Steering_inequality_main} is maximally violated by a state $\ket{\psi}_\mc{AB}\in\mathbb{C}^d\otimes\mc{H_B}$ and unitary $d$-outcome observables $A_k$ and $B_y$ ($k,y\in\{0,\ldots,d-1\}$) acting on, respectively, $\mathbb{C}^d$ and $\mc{H_B}$ such that the observables on Alice's trusted side are given by $A_k=XZ^{-k}$. Then, the following statement holds true: there exists a local unitary transformation on Bob's untrusted side, $U_\mc{B}: \mc{H_B}\to \mc{H_B}$, such that
\be\label{state_self_test_main}
(\I_\mc{A}\otimes U_\mc{B})\ket{\psi}_\mc{AB}=\ket{\phi_d^+}_\mc{AB}
\ee
and
\be\label{operator_self_test_main}
\forall y,\quad U_\mc{B}\mb{B}_y U_\mc{B}^\dagger=\overline{B}_y,
\ee
where $\mb{B}_y$ is $B_y$ projected onto the support of Bob's state $\rho_\mc{B}$ acting on $\mathbb{C}^d$ and $\overline{B}_y\coloneqq \overline{B}_y^{(1)}$ is given by \eqref{B(n)_main}.
\end{theorem}
\begin{proof} Here, we present only a sketch of the proof, which full version can be found in Appendix \ref{appendix_self-testing}. Note that we follow similar steps implemented in the proof of Theorem 1.1 of Ref. \cite{sarkar2022certification}. First, we prove that Bob's measurements split in $B_y=\mathbb{B}_y\oplus E_y$ where $\mb{B}_y\coloneqq \Pi_\mc{B} B_y \Pi_\mc{B}$ is the projection of $B_y$ onto the support of $\rho_\mc{B}$ and $E_i$ belongs to the complement of it. This is proved by implementing the projection $\Pi_\mc{B}$ onto \eqref{ABtilde_main} and using tricks to show that $\mathbb{B}_y$ is unitary. Then, we write $\ket{\psi}_\mc{AB}$ in its Schmidt decomposition form
\be\label{schmidt_decomposition_main}
\ket{\psi}_\mc{AB}=\sum_{i=0}^{d-1}\lambda_i\ket{u_i}\ket{v_i},
\ee
where $\{\ket{u_i}\}$ and $\{\ket{v_i}\}$ are orthonormal bases of $\mathbb{C}^d$ and $\mc{H}_\mc{B}$, respectively. Also, because $\rank(\rho_\mc{A})=d$, we have $\lambda_i> 0$, $\forall i$, such that $\sum_i \lambda_i^2=1$. Note that, there is a unitary transformation $U_\mc{B}$ that satisfies $U_\mc{B}\ket{v_i}=\ket{u_i^*}$, $\forall i$, where $^*$ denotes complex conjugation. Therefore, we can rewrite \eqref{schmidt_decomposition_main} as
\be
\ket{\psi}_\mc{AB}=\left(P_\mc{A}\otimes U_\mc{B}^\dagger \right)\ket{\phi_d^+},
\ee
where $P_\mc{A}$ is an operator that is diagonal in the basis $\{u_i\}$, with eigenvalues $\sqrt{d}\lambda_i$. From \eqref{ABtilde_main} and the above, it is possible to show that $P_\mc{A}$ is proportional to identity. Therefore, the above equation implies \eqref{state_self_test_main}. Again, from \eqref{ABtilde_main} and the fact that $P_\mc{A}$ is proportional to identity, we can also show that $U_\mc{B} \widetilde{\mb{B}}_k U_\mc{B}^\dagger =A_k^*$. From definition \eqref{substitution_rule_main}, that can be expressed as
\be
\sum_{y=0}^{d-1} \gamma_{yk}^{(1)} U_\mc{B} \mb{B}_y U_\mc{B}^\dagger =A_k^*.
\ee
We can now multiply both sides by $\left(\gamma_{y'k}^{(1)}\right)^*$, sum over $k$, and use \eqref{b_delta_main} to obtain
\be
\sum_{y=0}^{d-1} \delta_{yy'} U_\mc{B}\mb{B}_y U_\mc{B}^\dagger =\sum_{k=0}^{d-1} \left(\gamma_{y'k}^{(1)} \right)^* XZ^k.
\ee
Finally, from definition \eqref{B(n)_main}, the above equation results in
\be
U_\mc{B}\mb{B}_y U_\mc{B}^\dagger = \overline{B}_y,
\ee
which finishes the proof.
\end{proof}

Now, we are going to extend the above results to the case where we do not assume that (i) the measurements performed by Bob are projective and (ii) that the state shared by the parties is pure.

\begin{theorem}\label{Theorem_2_main}
    Assume that the steering inequality \eqref{Steering_inequality_main} is maximally violated by a state $\rho_\mc{AB}$ acting on $\mathbb{C}^d\otimes\mc{H_B}$ and unitary $d$-outcome observables $A_k$ and $B_y$ ($k,y\in\{0,\ldots,d-1\}$) acting on, respectively, $\mathbb{C}^d$ and $\mc{H_B}$ such that the observables on Alice's trusted side are given by $A_k=XZ^{-k}$. Then, the following three statements hold true:

    (i) Bob's measurements are projective, that is, all operators $B_{n|y}$ are unitary such that $B_{n|y}^d = \I$, 
    
    (ii) Bob's Hilbert space decomposes as $\mc{H_B}=(\mb{C}^d)_{\mc{B}'}\otimes \mc{H}_{\mc{B}''}$, and

    (iii) there exists a local unitary transformation on Bob's untrusted side, $U_\mc{B}: \mc{H_B}\to \mc{H_B}$, such that
\be\label{state_self_test2_main}
(\I_\mc{A}\otimes U_\mc{B})\rho_\mc{AB}(\I_\mc{A}\otimes U_\mc{B}^\dagger)=\ket{\phi_d^+}\bra{\phi_d^+}_\mc{AB'}\otimes\rho_{\mc{B}'}
\ee
and
\be\label{operator_self_test2_main}
\forall y,\quad U_\mc{B} B_y U_\mc{B}^\dagger=\overline{B}_y\otimes\I_\mc{B''},
\ee
where $\mc{B}''$ denotes Bob's auxiliary system and $\overline{B}_y\coloneqq \overline{B}_y^{(1)}$ is given by \eqref{B(n)_main}.
\end{theorem}
Before we present the proof of this theorem, let us stress again that the above statement demonstrates that all our steering inequalities are nontrivial as they enable self-testing maximally entangled states. Clearly, this would not be possible if the maximal quantum value of our functionals could be achieved by separable states.

\begin{proof} Here, we present only a sketch of the proof, which full version can be found in Appendix \ref{appendix_self-testing}. Note that we follow similar steps implemented in the proof of Theorem 1.2 of Ref. \cite{sarkar2022certification}. If we observe the maximal violation of inequality \eqref{Steering_inequality_main}, then we can infer from the SOS decomposition \eqref{SOS_decomposition_main} that
\be\label{ABtilde_rho=rho}
\left(A_k^n\otimes\widetilde{B}_k^{(n)}\right)\rho_\mc{AB}=\rho_\mc{AB}
\ee
for every $n\in\{1,\ldots,d-1\}$ and $k\in\{0,\ldots,d-1\}$. From \eqref{ABtilde_rho=rho} and \eqref{Bn^dagger=Bd-n_main}, we can do some manipulations using \eqref{b_delta_main} to prove that $\widetilde{B}_k^{(n)}$, as well as $B_{n|y}$, must be unitary. That means that Bob's measurements $N_{b|y}$ must be projective and, therefore, we can use $B_{n|y}=B_y^n$. To extend the proof of Theorem \ref{theorem_1_main} to mixed states, we start by writing $\rho_\mc{AB}$ in its Schmidt decomposition
\be\label{rho_decomposed}
\rho_\mc{AB}=\sum_{s=1}^kp_s\ket{\psi_s}\bra{\psi_s}_\mc{AB},
\ee
where $p_s\geqslant 0$ for every $s$ such that $\sum_s p_s=1$. If the steering inequality is violated by $\rho_\mc{AB}$, then every $\ket{\psi_s}$ must satisfy $(A_k^n\otimes\widetilde{B}_k^{(n)})\ket{\psi_s}_\mc{AB}=\ket{\psi_s}_\mc{AB}$. From Theorem \eqref{theorem_1_main}, we can conclude that there must be unitaries $U_\mc{B}^{(s)}$ such that
\begin{subequations}
    \begin{align}
        \left(\I_\mc{A}\otimes U_\mc{B}^{(s)}\right)\ket{\psi_s}_\mc{AB}&=\ket{\phi_d^+},\\
        U_\mc{B}^{(s)}\mathbb{B}_y^{(s)}\left(U_\mc{B}^{(s)}\right)^\dagger&=\overline{B}_y,
    \end{align}
\end{subequations}
for every $y$, where $\mathbb{B}_y^{(s)}\coloneqq\Pi_\mc{B}^s B_y \Pi_\mc{B}^s$ is the operator $B_y$ projected onto the support of Bob's local state $\rho_\mc{B}^{(s)}=\Tr_\mc{A}\left(\ket{\psi_s}\bra{\psi_s}_\mc{AB} \right)$. Similarly to the proof of Theorem \ref{theorem_1_main}, we can conclude that Bob's observables can be decomposed as
\be
B_y=\mathbb{B}_y^{(s)}\oplus E_y^{(s)},
\ee
where $\mathbb{B}_y^{(s)}$ are unitary operators such that $(\mathbb{B}_y^{(s)})^d=\I_d$ and $E_y^{(s)}$ are also unitaries acting on the complement of the local supports $V_s \equiv\text{supp}(\rho_\mc{B}^{(s)})$. Now, we can multiply the above equation by $\gamma_{yk}^{(1)}$ on both sides and take the summation over $y$ to obtain
    \be\label{Bk_direct_sum_main}
    \widetilde{B}_k=\widetilde{\mathbb{B}}_k^{(s)}\oplus \widetilde{E}_k^{(s)},
    \ee
where $\widetilde{\mathbb{B}}_k^{(s)}\coloneqq \sum_{y=0}^{d-1}\gamma_{yk}^{(1)}\mathbb{B}_y^{(s)}$, and similarly to $\widetilde{E}_k^{(s)}$. Because $\widetilde{B}_k=\widetilde{B}_k^{(1)}$ is unitary, we can use the same arguments from Theorem \ref{theorem_1_main} to show that $\widetilde{\mathbb{B}}_k^{(s)}$ is also unitary. From \eqref{Bk_direct_sum_main}, we can follow the same steps in the proof of Theorem 1.2 of Ref. \cite{sarkar2022certification} to conclude that the local supports $V_s$ are mutually orthogonal. That implies that Bob's local Hilbert space admits the decomposition $\mc{H_B}=V_1\oplus V_2\oplus \ldots\oplus V_{k}=(\mathbb{C}^d)_{\mc{B}'}\otimes \mc{H}_{\mc{B}''}$, where the second equality is due to the fact that $\dim V_s=d$ for every $s=1,\ldots,K$. It is possible to show that the local supports $V_s$ are orthogonal subspaces. Therefore, there must exist a single unitary $U_\mc{B}$ that transform $\ket{\psi_s}_\mc{AB'}$ into $\ket{\phi_d^+}$ plus some auxiliary pure state $\ket{s}_\mc{B''}$. Applying $U_\mc{B}$ to \eqref{rho_decomposed} provides \eqref{state_self_test2_main}. From the decomposition of $\mc{H_B}$, the operators $\widetilde{B}_k$ under the action of $U_\mc{B}$ must admit the decomposition
    \be\label{U_B_decomposed_main}
    U_\mc{B} \widetilde{B}_k U_\mc{B}^\dagger= \sum_{i,j}\widetilde{B}_{i,j,k}\otimes \ket{i}\bra{j}_\mc{B''}, 
    \ee
where $\widetilde{B}_{i,j,k}$ are matrices of dimension $d$ acting over $(\mathbb{C}^d)_\mc{B'}$. From \eqref{ABtilde_rho=rho} and the above, it can be show that $\widetilde{B}_{i,j,k}=\delta_{ij}A_k^*$. Therefore, we have $U_\mc{B}\widetilde{B}_k U_\mc{B}^\dagger=A_k^*\otimes\I_\mc{B''}$, which implies \eqref{operator_self_test2_main} as we wanted to prove.
\end{proof}

\subsection{Classical bound for the qutrit case}\label{sec_classical_bound}

To calculate the classical bound of \eqref{Steering_inequality_main}, first we need to provide solutions to \eqref{b_delta_main}. In turns out that the system of equations
can be solved analytically and all possible solutions can be brought to two 
possibilities, one of them being the following one,
\be
\begin{array}{c}
y=1:\left\{\begin{array}{l}
\phi_{0,1}=\phi_{0,0}-\dfrac{2\pi}{3},\\[1ex]
\phi_{1,1}=\phi_{1,0},
\end{array}\right.\\
y=2:
\left\{
\begin{array}{l}
\phi_{0,2}=\phi_{0,0}-\dfrac{2\pi}{3},\\[1ex]
\phi_{1,2}=\phi_{1,0}+\dfrac{2\pi}{3}.
\end{array}\right.
\end{array}
\label{conditions}
\ee
which on the level of observables translates to 
\begin{subequations}
\begin{align}\label{1}
\overline{B}_0 &\equiv\begin{pmatrix}
 0 & 0 & e^{\mathbbm{i} \phi _{0,0}} \\
 e^{\mathbbm{i} \phi _{1,0}} & 0 & 0 \\
 0 & e^{-\mathbbm{i} \left(\phi _{0,0}+\phi _{1,0}\right)} & 0 
\end{pmatrix},\\
\label{2}\overline{B}_1 &\equiv\begin{pmatrix}
 0 & 0 & \omega^2 e^{\mathbbm{i} \phi _{0,0}} \\
 e^{\mathbbm{i} \phi _{1,0}} & 0 & 0 \\
 0 & \omega e^{-\mathbbm{i} \left(\phi _{0,0}+\phi _{1,0}\right)} & 0 
\end{pmatrix},\\
\label{3}\overline{B}_2 &\equiv\begin{pmatrix}
 0 & 0 & \omega^2 e^{\mathbbm{i} \phi _{0,0}} \\
 \omega e^{\mathbbm{i} \phi _{1,0}} & 0 & 0 \\
 0 & e^{-\mathbbm{i} \left(\phi _{0,0}+\phi _{1,0}\right)} & 0 
\end{pmatrix},
\end{align}
\end{subequations}
where $\phi_{0,0},\phi_{1,0}\in\mathbb{R}$. The other solution is presented in Appendix \ref{appendix_Uniqueness}.
%
%
%

Once the numbers $\{\phi_{j,y}\}$ are specified, the classical bound can be numerically calculated as
\be\label{classical_bound}
\beta_C\left(\phi_{0,0},\phi_{1,0}\right)=\max_{\{a_0,a_1,a_2,b_0,b_1,b_2\}}\sum_{n=1}^2\sum_{k,y=0}^{2}\gamma_{yk}^{(n)}a_k^n b_y^n,
\ee
where $a_k,b_y\in\{1,\omega,\omega^2\}$, for every $k$ and $y$, are the outcomes of the measurements. 
%
In Fig. \ref{fig:plot} we plot the numerical evaluation of \eqref{classical_bound} as a function of $\{\phi_{0,0},\phi_{1,0}\}$ following solution \eqref{conditions}. 


For the specific values $\{\phi_{0,0},\phi_{1,0}\}=\{4\pi/9,-2\pi/9\}$, we regain the steering inequality derived in \cite{kaniewski2019maximal} for $d=3$ with classical bound $\beta_C=6\cos(\pi/9)\sim 5.63816$. 

Let us finally mention that finding analytically a general solution to the system \eqref{b_delta_main} is a difficult task and, in fact, no general method is known that does the job. Yet, a particular solution to these equations for prime $d$ was found in Ref. \cite{kaniewski2019maximal} and reads 
\begin{equation}
    \gamma_{yk}^{(n)}=\frac{\lambda_n^*}{\sqrt{d}}\omega^{-nyk}\omega^{-k(k+1)},
\end{equation}
where $\lambda_n$ is a complex coefficient that has a quite complicated structure and is explicitly defined in Eqs. (18), (19) and (20) of Ref. \cite{kaniewski2019maximal}. However, to find other solutions one can use
numerical methods.

\begin{figure}[h]
    \centering
    \includegraphics[width=\linewidth]{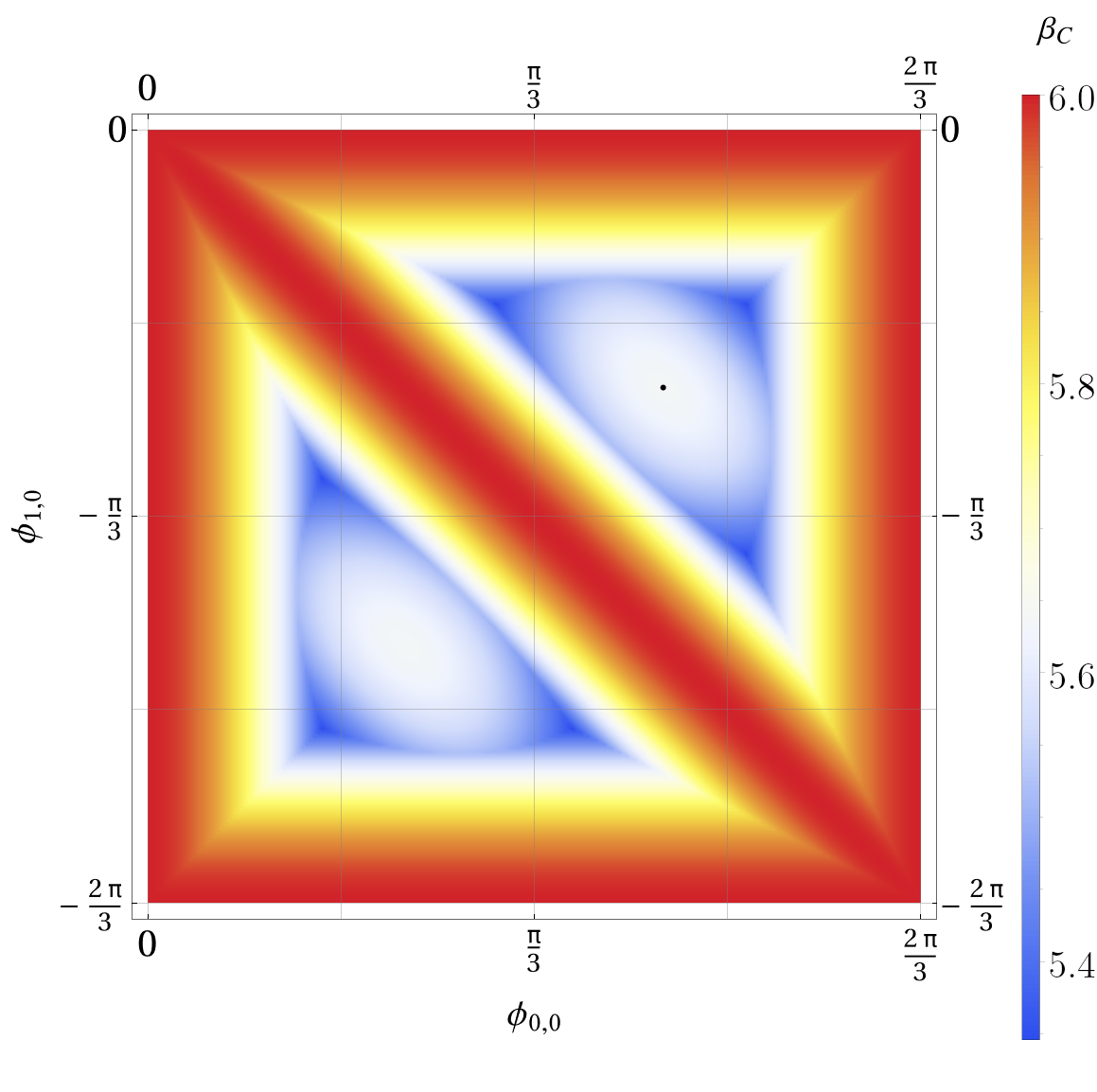}
    \caption{Classical bound as a function of $\phi_{0,0}$ and $\phi_{1,0}$ from solution \eqref{conditions}. Values range from $\sim 5.34604$ to $6$. The black dot represents the classical bound $\beta_C=6\cos(\pi/9)\sim 5.63816$ of the Steering inequality derived in Ref. \cite{kaniewski2019maximal} when $\{\phi_{0,0},\phi_{1,0}\}=\{4\pi/9,-2\pi/9\}$.}
    \label{fig:plot}
\end{figure}

\section{Robust self-testing}\label{sec_robustness}

If the maximal violation of the steering inequality \eqref{Steering_inequality_main} is not reached by an $\epsilon$, then the state and the measurements performed in the experiment will be close to the reference state and measurements by a factor as a function of $\epsilon$. Below, we provide the robustness analysis of our self-testing statement for the qutrit case, assuming that the measurements performed by Bob are projective and the state is pure. 
\begin{theorem}
For $d=3$, consider the unitary $d$-outcome observables $A_k=XZ^{-k}$ and $B_y$, for $k,y\in\{0,1,2\}$ acting on, respectively $\mathbb{C}^3$ and $\mc{H}_\mc{B}$. If the Steering inequality \eqref{Steering_inequality_main}
\be
\langle\mc{S}\rangle=\sum_{n=1}^{2}\sum_{k=0}^{2} \langle A_k^n\otimes\widetilde{B}_k^{(n)} \rangle \leqslant \beta_C,
\ee
is violated by a state $\ket{\psi}_\mc{AB} \in \mathbb{C}^{3} \otimes \mathcal{H}_{B} $ and observables $ B_{y} $ such that $\langle\mc{S}\rangle \geqslant 6 - \epsilon > \beta_C$, then there exists a unitary operation $ U_\mc{B}: \mathcal{H}_\mc{B} \to \mathcal{H}_\mc{B} $ such that
\begin{equation}\label{bound1}
    \left\|\left(\I_\mc{A}\otimes U_\mc{B}\right)\left( \I_\mc{A}\otimes B_y \right)\ket{\psi}-\I\otimes \overline{B}_y\ket{\phi_3^+}\right\|\leqslant 2\sqrt{3}(2\epsilon)^\frac{1}{4}
\end{equation}
and
\begin{equation}\label{bound2}
    \left\|B_y-\overline{B}_y\right\|_2\leqslant 12\left(2\epsilon \right)^\frac{1}{4},
\end{equation}
where $ k = 0, 1, 2 $ and $ \overline{B}_{y}$ are Bob's ideal observables given by \eqref{Bbar_main}, and $ \| \cdot \|_{2} $ stands for the Hilbert-Schmidt norm.
\end{theorem}
We left the proof of the robustness in Appendix \ref{robustness}. Our robustness result for $d=3$ already well illustrates the behaviour of error bounds as a function of $\epsilon$. The generalisation to higher $d$ will introduce some additional $d$-dependent coefficient while keeping similar behaviour in $\epsilon$.

\section{Discussion}\label{sec_discussion}

Recently, very general schemes for device-independent certification of quantum measurements have been introduced \cite{chen2023all,sarkar2023universal}. However, in these schemes, the generality comes at the price of considerably increasing the cost of implementing them in terms of the number of measurements that the parties must perform. One way to reduce that cost is to resort to semi-device-independent scenarios, such as that based on quantum steering, where it is assumed that one of the measuring devices is trusted and performs known measurements.

In fact, by making the scheme SDI, we can reduce the number of measurements required for the certification but still encompass a large class of measurement operators. Here, we constructed a family of steering inequalities \eqref{Steering_inequality_main} that are tailored to certain classes of projective $d$-outcome measurements on the untrusted side and a set of $d$ fixed and known measurements on the trusted one. 
One can think of our inequalities as a generalization of the Bell inequalities introduced in Ref. \cite{kaniewski2019maximal} to more general classes of measurements, which comes however at the cost of assuming that one of the measuring devices is trusted; this is illustrated on Fig. \ref{fig:plot} for a particular case of $d=3$. It is important to note that for a measure-zero set of points in Fig. \ref{fig:plot}, the classical bound reaches the value 6, which is also the quantum bound in dimension 3. In such a case, the steering inequality is trivial and, therefore, useless for self-testing. 

We then demonstrate that the new steering inequalities can be used for certification purposes and show that their maximal violation allows for SDI certification of the considered classes of measurements on the untrusted side as well as the maximally entangled state of two qudits.


An important open question concerns reducing the number of measurements from $d$ to 2, which is in fact the minimal number of measurements necessary to observe Bell-nonlocality or quantum steering. Finally, it would also be useful to find a way of extending our inequalities to the device-independent scenario.

\section{Acknowledgments}
This project was funded within the QuantERA II Programme (VERIqTAS project) that has received funding from the European Union's Horizon 2020 research and innovation programme under Grant Agreement No 101017733 and from the Polish National Science Center (projects No. 2021/03/Y/ST2/00175 and No. 2022/46/E/ST2/00115).
This project has received funding from the European Union's Horizon Europe research and innovation programme under grant agreement No 101080086NeQST.

\bibliography{bibliography.bib}

\appendix

\onecolumngrid


\setcounter{theorem}{0}

\section{Construction of the steering operator}\label{appendix_Steering_ope}
Let us consider the Steering scenario with two qudits of local dimension $d$, where Alice performs $d$ measurements $A_x=XZ^{-x}$, for $x=0,\ldots,d-1$, and Bob performs $d$ arbitrary measurements. The measurements that we want to self-test on Bob's side have the following form
\be\label{Bbar}
\overline{B}_y=\dfrac{1}{d}\sum_{k,j=0}^{d-1}e^{ \bbm{i}\phi_{j\oplus 1,y}}\omega^{-kj}XZ^k,
\ee
where $\phi_{j,y}\in\mathbb{R}$ for any $j,y=0,\ldots,d-1$ such that
\be\label{sum_zero}
\phi_{d-1,y}=-\sum_{j=0}^{d-2}\phi_{j,y},\qquad \text{ for every } y=0,\ldots,d-1.
\ee
It can be checked that $\overline{B}_y$ is unitary and the above condition guarantees that $\overline{B}_y^d=\I$. Throughout the paper, the symbol $\oplus$ means addition modulo $d$. For $d=3$, for instance, we want to certify three measurements made by Bob identified by $y=0,1,2$ that depend on two free real parameters $\phi_{0,y}$ and $\phi_{1,y}$, for each $y$, in the following form:
\be
\overline{B}_y=
\begin{pmatrix}
 0 & 0 & e^{i \phi _{0,y}} \\
 e^{i \phi _{1,y}} & 0 & 0 \\
 0 & e^{-i \left(\phi _{0,y}+\phi _{1,y}\right)} & 0     
\end{pmatrix}, \qquad \text{for }d=3.
\ee
As we are going to see, it will be necessary to use the powers of $\overline{B}_y$. To do that, let us implement the explicit forms for $X$ and $Z$ in the above expression:
\be
    \overline{B}_y = \frac{1}{d}\sum_{k,j,l=0}^{d-1} e^{ \bbm{i}\phi_{j\oplus 1,y}}\omega^{-kj} \omega^{lk}\ket{l+1}\bra{l}.
\ee
The sums over $k$ and $j$ result in
\be
    \frac{1}{d}\sum_{k,j=0}^{d-1} e^{ \bbm{i}\phi_{j\oplus 1,y}}\omega^{-kj} \omega^{lk} = \frac{1}{d}\sum_{k,j=0}^{d-1} e^{ \bbm{i}\phi_{j\oplus 1,y}}\omega^{k(l-j)}
    = \sum_{j=0}^{d-1} e^{ \bbm{i}\phi_{j\oplus 1,y}}\delta_{l,j}=e^{\bbm{i}\phi_{l\oplus 1,y}}.
\ee
Therefore, we can write $\overline{B}_y$ as
\be
\overline{B}_y=\sum_{l=0}^{d-1}e^{\bbm{i}\phi_{l\oplus 1,y}}\ket{l+1}\bra{l}.
\ee
Now, let us explicitly calculate the $n=1,\ldots,d-1$ powers of $\overline{B}_y$:
\begin{subequations}
    \begin{align}
    \overline{B}_y^n &=\sum_{l_1,\ldots,l_n=0}^{d-1}e^{\bbm{i}(\phi_{l_1\oplus 1,y}+\ldots+\phi_{l_n\oplus 1,y})} \ket{l_1+1}\bra{l_1}\ldots\ket{l_n+1}\bra{l_n},\\
    &= \sum_{l=0}^{d-1}e^{\bbm{i}(\phi_{l\oplus n,y}+\phi_{l\oplus n-1,y}+\phi_{l\oplus n-2,y}+\ldots+\phi_{l\oplus 1,y})}\ket{l+n}\bra{l},\\
    &= \sum_{l=0}^{d-1}e^{\bbm{i}\sum_{m-1}^{n}\phi_{l\oplus m,y}}\ket{l+n}\bra{l}.\label{Bn_3}
    \end{align}
\end{subequations}
Inspired by the fact that \cite{kaniewski2019maximal}
\be\label{XZkn}
\left(XZ^k \right)^n=\omega^{kn(n-1)/2}\sum_{l'=0}^{d-1}\omega^{nkl'}\ket{l'+n}\bra{l'},
\ee
we can define the following set of operators
\be\label{Bbarn}
\overline{B}_y^{(n)}\coloneqq\dfrac{1}{d}\sum_{k,l=0}^{d-1}e^{\bbm{i}\sum_{m=1}^{n}\phi_{l\oplus m,y}}\omega^{-knl}\omega^{-kn(n-1)/2}\left(XZ^k \right)^n,
\ee
for $n=1,\ldots,d-1$. The direct substitution of \eqref{XZkn} into \eqref{Bbarn} results in \eqref{Bn_3}, that is,
\be
\overline{B}_y^{(n)}=\overline{B}_y^n, \qquad \text{for }n=1,\ldots,d-1,
\ee
which satisfies
\be
\left(\overline{B}_y^{(n)}\right)^d=\left(\overline{B}_y^n\right)^d=\I, \qquad \text{for }n=1,\ldots,d-1.
\ee
Also, note that if we take $n=d$ in \eqref{Bbarn}, we obtain $\overline{B}_y^{(d)}=d\I$. In addition, because $\overline{B}_y$ is unitary, its powers are also unitary, as well as, $\overline{B}_y^{(n)}$ for every $n=1,\ldots,d-1$.

To make the equations clearer, let us define
\be\label{byk}
\gamma_{yk}^{(n)}\coloneqq \frac{1}{d}\omega^{kn(n-1)/2}\sum_{l=0}^{d-1}e^{-\bbm{i}\sum_{m=1}^{n}\phi_{l\oplus m,y}}\omega^{knl},
\ee
which means that we can rewrite $\overline{B}_y^{(n)}$ as
\be\label{B(n)}
\overline{B}_y^{(n)}=\sum_{k=0}^{d-1}\left(\gamma_{yk}^{(n)}\right)^* \left(XZ^k\right)^n.
\ee

Before proceeding to the construction of the steering operator, let us prove some facts about the above operator.

\begin{fact}\label{fact_XZkn}
For any pair of integers $k$ and $n$, and for any $d\geqslant 3$, the operator $\left(XZ^k\right)^n$ satisfies the following relation:
\be
\left(XZ^k \right)^{d-n}=\left[\left(XZ^k \right)^n \right]^\dagger.
\ee
\end{fact}
\begin{proof}
    From Eq. \eqref{XZkn}, we have
    \be\label{XZk^d-n}
    \left(XZ^k \right)^{d-n}=\sum_{l'=0}^{d-1}\omega^{k(d-n)(d-n-1)/2}\omega^{k(d-n)l'}\ket{l'+d-n}\bra{l'},
    \ee
    With some algebra, we can show that $\omega^{k(d-n)(d-n-1)/2}=\omega^{kn(n+1)/2}$ and that $\omega^{k(d-n)l'}=\omega^{-knl'}$. By applying these relations to \eqref{XZk^d-n}, we have
    \be
    \left(XZ^k \right)^{d-n}=\omega^{kn(n+1)/2}\sum_{l'=0}^{d-1}\omega^{-nkl'}\ket{l'-n}\bra{l'},    
    \ee
    where we have used the fact that $\ket{a+d}\equiv\ket{a}$ for any integer $a$. Now, we can take the conjugate transpose of the above equation by doing
    \be\label{XZk_conjugate}
    \left[\left(XZ^k \right)^{d-n}\right]^\dagger=\omega^{-kn(n+1)/2}\sum_{l'=0}^{d-1}\omega^{nkl'}\ket{l'}\bra{l'-n}.
    \ee
    Observe that we can rearrange the terms inside the above summation in the following way:
    \begin{subequations}
        \begin{align}
    \sum_{l'=0}^{d-1}\omega^{nkl'}\ket{l'}\bra{l'-n}&=\ket{0}\bra{-n}+\omega^{nk}\ket{1}\bra{1-n}+\ldots+\omega^{nkn}\ket{n}\bra{0}+\ldots+\omega^{nk(d-1)}\ket{d-1}\bra{d-1-n},\\
    &=\omega^{nkn}\ket{n}\bra{0}+\omega^{nk(n+1)}\ket{n+1}\bra{1}+\ldots+\omega^{nk(n+d-1)}\ket{n+d-1}\bra{d-1},\\
    &=\sum_{l=0}^{d-1}\omega^{nk(n+l)}\ket{l+n}\bra{l}.
    \end{align}
    \end{subequations}
    By using the above in \eqref{XZk_conjugate}, we obtain
    \be
    \left[\left(XZ^k \right)^{d-n}\right]^\dagger=\omega^{-kn(n+1)/2}\sum_{l=0}^{d-1}\omega^{nk(l+n)}\ket{l+n}\bra{l}.
    \ee
    After some basic algebra, we obtain the desired result:
    \be
    \left(XZ^k \right)^{d-n}=\left[\left(XZ^k \right)^{n}\right]^\dagger
    \ee
\end{proof}

\begin{fact}\label{fact_b}
For every $y,k=0,\ldots,d-1$, every $n=1,\ldots,d-1$, and any dimension $d\geqslant 3$, the coefficients \eqref{byk} satisfy the relation
    \be
    \gamma_{yk}^{(d-n)}=\left(\gamma_{yk}^{(n)} \right)^*.
    \ee
\end{fact}
\begin{proof}
From definition \eqref{byk}, we can directly calculate
    \be
    \gamma_{yk}^{(d-1)}= \frac{1}{d}\omega^{k(d-n)(d-n-1)/2}\sum_{l=0}^{d-1}e^{-\bbm{i}\sum_{m=1}^{d-n}\phi_{l\oplus m,y}}\omega^{k(d-n)l}.
    \ee
    Similarly to the proof of Fact \ref{fact_XZkn}, we have
    \be\label{b_yk^d-1*}
    \left(\gamma_{yk}^{(d-1)}\right)^*= \frac{1}{d}\omega^{-kn(n+1)/2}\sum_{l=0}^{d-1}e^{\bbm{i}\sum_{m=1}^{d-n}\phi_{l\oplus m,y}}\omega^{knl}.
    \ee
    Now, observe that we can rewrite the summation inside the exponential by summing zero in the following way:
    \be
    \sum_{m=1}^{d-1}\phi_{l\oplus m,y}=\sum_{m=1}^{d-1}\phi_{l\oplus m,y}+\sum_{m=d-n+1}^{d}\phi_{l\oplus m,y}-\sum_{m=d-n+1}^{d}\phi_{l\oplus m,y}=-\sum_{m=d-n+1}^{d}\phi_{l\oplus m,y}=-\sum_{m'=1}^{n}\phi_{l\oplus m'+d-n,y},
    \ee
    where we implemented a change of variables in the last equality given by $m=m'+d-n$. Again, similarly to the proof of Fact \ref{fact_XZkn}, we can use the above equation in \eqref{b_yk^d-1*} and rewrite the sum over $l$ as
    \be
    \sum_{l=0}^{d-1}\exp\left(-\bbm{i}\sum_{m'=1}^{n}\phi_{l\oplus m'+d-n,y} \right)\omega^{knl}=\sum_{l=0}^{d-1}\exp\left(-\bbm{i}\sum_{m'=1}^{n}\phi_{l\oplus m',y} \right)\omega^{kn(l+n)},
    \ee
    which gives us
    \be
    \left(\gamma_{yk}^{(d-1)}\right)^*= \frac{1}{d}\omega^{-kn(n+1)/2}\sum_{l=0}^{d-1}\exp\left(-\bbm{i}\sum_{m'=1}^{n}\phi_{l\oplus m',y} \right)\omega^{kn(l+n)}.
    \ee
    By using some basic algebra in the above equation, we obtain the desired result.
\end{proof}

\begin{fact}
The operators $\overline{B}_y^{(n)}$ given by \eqref{Bbarn} satisfy the relation
\be
\overline{B}_y^{(d-n)}=\left(\overline{B}_y^{(n)}\right)^\dagger, \qquad\text{for }n=1,\ldots,d-1.
\ee
\end{fact}
\begin{proof}
    Directly from Facts \ref{fact_XZkn} and \ref{fact_b}.
\end{proof}

Now, let us construct the Steering operator. To do that, we need to impose a condition on the coefficients $\gamma_{yk}^{(n)}$, which is going to restrict the range of values for $\{\phi_{j,y}\}$. We need to impose that 
\be\label{b_delta}
\sum_{k=0}^{d-1} \left(\gamma_{y'k}^{(n)}\right)^*\gamma_{yk}^{(n)}=\delta_{yy'},
\ee
for $n=1,\ldots,d-1$. Because
\be
\sum_{k=0}^{d-1} \left(\gamma_{y'k}^{(n)}\right)^*\gamma_{yk}^{(n)}=\dfrac{1}{d^2}\sum_{k,l,l'=0}^{d-1} \omega^{kn(l-l')} e^{-\bbm{i}\sum_{m=1}^{n}\phi_{l\oplus m,y}}e^{\bbm{i}\sum_{m=1}^{n}\phi_{l'\oplus m,y'}}=\dfrac{1}{d}\sum_{l=0}^{d-1}e^{\bbm{i}\sum_{m=1}^{n}\phi_{l\oplus m,y'}-\phi_{l\oplus m,y}},
\ee
relation \eqref{b_delta} is true \textit{iff} the following system holds:
\begin{subequations}\label{system}
    \begin{align}
    &\sum_{l=0}^{d-1}\exp\left(\bbm{i}\sum_{m=1}^{n}\phi_{l\oplus m,y'}-\phi_{l\oplus m,y} \right)=d\delta_{yy'},\qquad\text{for every pair }y,y'=0,\ldots,d-1\text{ and every }n=1,\ldots,d-1;\\
    &\sum_{l=0}^{d-1}\phi_{l,y}=0,\qquad\text{for every }y=0,\ldots,d-1.
\end{align}
\end{subequations}
For instance, for $d=3$, the above system establishes that Bob's reference measurements $\overline{B}_1$ and $\overline{B}_{2}$ must depend on $\overline{B}_0$, which is written as a function of $2$ completely free parameters, namely $\phi_{0,0}$ and $\phi_{1,0}$, since $\phi_{2,0}=-(\phi_{0,0}+\phi_{1,0})$. Later on, we are going to provide some solutions for the above system.

If \eqref{system} holds true, then the following relation also holds:
\be
\sum_{y=0}^{d-1}\exp\left(\bbm{i}\sum_{m=1}^{n}\phi_{l'\oplus m,y}-\phi_{l\oplus m,y} \right)=d\delta_{ll'},\qquad\text{for every pair }l,l'=0,\ldots,d-1\text{ and every }n=1,\ldots,d-1,
\ee
which means that
\be\label{b_delta_kk'}
\sum_{y=0}^{d-1}\left(\gamma_{yk'}^{(n)}\right)^* \gamma_{yk}^{(n)}=\delta_{kk'}.
\ee

From the above, we can multiply \eqref{B(n)} by $\gamma_{yk'}^{(n)}$ and sum over $y$ on both sides to obtain
\be
\left(XZ^k\right)^n=\sum_{y=0}^{d-1}\gamma_{yk}^{(n)}\overline{B}_y^{(n)}.
\ee
Now, we can specify the substitution rule
\be\label{substitution_rule}
\left(XZ^k\right)^n\longrightarrow \widetilde{B}_k^{(n)}\coloneqq\sum_{y=0}^{d-1}\gamma_{yk}^{(n)}B_{n|y},
\ee
where $B_y=\{B_{n|y}\}$, for $y=0,\ldots,d-1$ and $n=1,\ldots,d-1$, represent the $d$-outcome arbitrary measurements satisfying
\be\label{Bn^dagger=Bd-n}
B_{n|y}^\dagger=B_{d-n|y}=B_{-n|y}
\ee
and
\be\label{BB<1}
B_{n|y}^\dagger B_{n|y}\leqslant \I.
\ee
Now, we can consider a Steering operator where Alice performs $d$ known measurements $A_k=XZ^{-k}$, for $k=0,\ldots,d-1$, and Bob performs $d$ arbitrary measurements $B_y$, for $y=0,\ldots,d-1$. The Steering operator reads
\be\label{Steering_operator}
\mc{S}=\sum_{n=1}^{d-1}\sum_{k=0}^{d-1} A_k^n\otimes\widetilde{B}_k^{(n)},
\ee
which can be written as
\be
\mc{S}=\sum_{n=1}^{d-1}\sum_{k,y=0}^{d-1}\left[  \left(XZ^{-k}\right)^n\otimes \gamma_{yk}^{(n)}B_{n|y} \right].
\ee
Therefore, we can propose a Steering inequality as
\be\label{Steering_inequality}
\sum_{n=1}^{d-1}\sum_{k=0}^{d-1} \langle A_k^n\otimes\widetilde{B}_k^{(n)} \rangle \leqslant \beta_C,
\ee
where $\beta_C$ is the maximum value achieved by classical strategies, henceforth called \textit{classical bound}.

\section{Quantum Bound}\label{appendix_QB}
To calculate the quantum bound of operator $\mc{S}$, let us use the following SOS decomposition:
\begin{align}
\sum_{n=1}^{d-1}\sum_{k=0}^{d-1}\left[\I- A_k^n\otimes\widetilde{B}_k^{(n)}\right]^\dagger\left[\I- A_k^n\otimes\widetilde{B}_k^{(n)}\right]&=d(d-1)\I-\mc{S}-\mc{S}^\dagger+\I\otimes\sum_{n=1}^{d-1}\sum_{k=0}^{d-1} \left(\widetilde{B}_k^{(n)}\right)^\dagger\widetilde{B}_k^{(n)},\label{SOS}\\
&=d(d-1)\I-2\mc{S}+\I\otimes\sum_{n=1}^{d-1}\sum_{y,y'=0}^{d-1}\left[\sum_{k=0}^{d-1} \left(\gamma_{y'k}^{(n)}\right)^*\gamma_{yk}^{(n)} \right]B_{n|y'}^\dagger B_{n|y},
\end{align}
where we have used Fact \ref{fact_b} and Eq. \eqref{Bn^dagger=Bd-n} to show that the Steering operator \eqref{Steering_operator} is Hermitian, i.e., $\mc{S}=\mc{S}^\dagger$. Now we can use \eqref{b_delta} and the fact that $B_{n|y}^\dagger B_{n|y}\leqslant\I$ to obtain
\begin{align}\label{SOS_decomposition}
0 \leqslant \dfrac{1}{2}\sum_{n=1}^{d-1}\sum_{k=0}^{d-1}\left[\I- A_k^n\otimes\widetilde{B}_k^{(n)}\right]^\dagger\left[\I- A_k^n\otimes\widetilde{B}_k^{(n)}\right]&\leqslant d(d-1)\I-\mc{S},
\end{align}
which implies that the maximum value achieved by $\mc{S}$ is given by $\beta_Q\coloneqq d(d-1)$. If Alice and Bob perform the measurements represented by
\be
A_x=XZ^{-x}\qquad\text{and}\qquad B_{n|y}=\sum_{k=0}^{d-1}\left(\gamma_{yk}^{(n)}\right)^* \left(XZ^k\right)^n,
\ee
for $x=0,\ldots,d-1$ and $y=0,\ldots,d-1$ respectively, then operator \eqref{Steering_operator} becomes a sum of stabilizing operators of the state $\ket{\phi_d^+}$, that is,
\be
\left(XZ^{-k}\right)^n\otimes \left(XZ^{k}\right)^n \ket{\phi_d^+}=\ket{\phi_d^+}. 
\ee

\section{Self-testing}\label{appendix_self-testing}

We start now by developing the self-testing statements regarding our Steering operator considering that the state $\rho_\mc{AB}=\ket{\psi}\bra{\psi}_\mc{AB}$ is pure and the measurements $B_y$ performed by Bob are projective. Because of that, each of those measurements can be encoded in a single unitary observable $B_y\coloneqq B_{1|y}$ that satisfies $B_y^d=\I$. Consequently, they also satisfy $B_{n|y}=B_y^n$ and
\be
B_y^{d-n}=B_y^{-n}=\left(B_y^n\right)^\dagger.
\ee
From the above and from Fact \ref{fact_b}, we can conclude that the operators $\widetilde{B}_k^{(n)}$ defined in \eqref{substitution_rule} must satisfy
\be\label{Btilde^d-n}
\widetilde{B}_k^{(d-n)}=\left(\widetilde{B}_k^{(n)}\right)^\dagger.
\ee
Now, consider the state and the measurements that give the maximum violation of \eqref{Steering_operator}. We can infer from the SOS decomposition \eqref{SOS_decomposition} that
\be\label{AB_tilde}
A_k^n\otimes \widetilde{B}_k^{(n)}\ket{\psi}_\mc{AB}=\ket{\psi}_\mc{AB}
\ee
for every $k=0,\ldots,d-1$ and $n=1,\ldots,d-1$. By implementing the notation $\widetilde{B}_k\coloneqq \widetilde{B}_k^{(1)}$, we can write
\be\label{ABtilde}
A_k\otimes \widetilde{B}_k\ket{\psi}_\mc{AB}=\ket{\psi}_\mc{AB}.
\ee

The self-testing theorem can be stated as follows:

\begin{theorem}\label{theorem_1}
Assume that the Steering inequality \eqref{Steering_inequality} is maximally violated by a state $\ket{\psi}_\mc{AB}\in\mathbb{C}^d\otimes\mc{H_B}$ and unitary $d$-outcome observables $A_k$ and $B_y$ ($k,y\in\{0,\ldots,d-1\}$) acting on, respectively, $\mathbb{C}^d$ and $\mc{H_B}$ such that the observables on Alice's trusted side are given by $A_k=XZ^{-k}$. Then, the following statement holds true: there exists a local unitary transformation on Bob's untrusted side, $U_\mc{B}: \mc{H_B}\to \mc{H_B}$, such that
\be\label{state_self_test}
(\I_\mc{A}\otimes U_\mc{B})\ket{\psi}_\mc{AB}=\ket{\phi_d^+}_\mc{AB}
\ee
and
\be\label{operator_self_test}
\forall y,\quad U_\mc{B}\mb{B}_y U_\mc{B}^\dagger=\overline{B}_y,
\ee
where $\mb{B}_y$ is $B_y$ projected onto the support of Bob's state $\rho_\mc{B}$ acting on $\mathbb{C}^d$ and $\overline{B}_y\coloneqq \overline{B}_y^{(1)}$ are given by \eqref{Bbarn}.
\end{theorem}
\begin{proof}
(\textit{Step 1.}) First, let us show that Bob's measurements are split into a direct sum as
\be
B_y=\mb{B}_y\oplus E_y,
\ee
where $\mb{B}_y\coloneqq \Pi_\mc{B} B_y \Pi_\mc{B}$ is defined in the support of $\rho_\mc{B}$ and $E_i$ belongs to the complement of it. Indeed, we can write $B_y$ in a block form such as
\be\label{block}
B_y=\begin{pmatrix}
    \mb{B}_y & C_y \\
    D_y & E_y
\end{pmatrix},
\ee
where $E_y\coloneqq\Pi_\mc{B}^{\perp} B_y \Pi_\mc{B}^{\perp}$ is the projection onto the complement of the support of $\rho_\mc{B}$. We can use $\Pi_\mc{B}$ onto \eqref{AB_tilde} to obtain
\be\label{Ak_Bbbk}
A_k\otimes \Pi_\mc{B} \widetilde{B}_k \Pi_\mc{B}\ket{\psi}_\mc{AB}= A_k\otimes \widetilde{\mb{B}}_k \ket{\psi}_\mc{AB}=\ket{\psi}_\mc{AB},
\ee
where
\be\label{Bk_bold}
\widetilde{\mb{B}}_k\coloneqq \Pi_\mc{B} \widetilde{B}_k \Pi_\mc{B}=\sum_{y=0}^{d-1} \gamma_{yk}^{(1)} \Pi_\mc{B} B_y \Pi_\mc{B}=\sum_{y=0}^{d-1} \gamma_{yk}^{(1)} \mb{B}_y.
\ee
The Eq. \eqref{Btilde^d-n} is valid, in particular, for $n=1$, which together with \eqref{AB_tilde} for $n=d-1$ provides us
\be\label{ABdagger}
A_k^\dagger\otimes\widetilde{\mb{B}}_k^\dagger\ket{\psi}_\mc{AB}=\ket{\psi}_\mc{AB}.
\ee
Now, let us insert \eqref{Ak_Bbbk} in the l.h.s. of \eqref{ABdagger} to obtain
\be
\I_d\otimes\widetilde{\mb{B}}_k^\dagger\widetilde{\mb{B}}_k\ket{\psi}_\mc{AB}=\ket{\psi}_\mc{AB}.
\ee
Because the reduced density matrix associated with subsystem $\mc{A}$ of $\ket{\psi}_\mc{AB}$ is full rank, the above equation is equivalent to
\be
\widetilde{\mb{B}}_k^\dagger\widetilde{\mb{B}}_k=\I_d.
\ee
Similar arguments result in $\widetilde{\mb{B}}_k\widetilde{\mb{B}}_k^\dagger=\I_d$, which proves that $\widetilde{\mb{B}}_k$ is unitary. Finally, by applying \eqref{Ak_Bbbk} to itself $d$ times, we obtain $\widetilde{\mb{B}}_k^d=\I_d$, which means that $\widetilde{\mb{B}}_k$, for every $k=0,\ldots,d-1$, is a proper quantum observable.

Additionally, we can prove that $\mb{B}_y$ is also unitary. Let us multiply \eqref{Ak_Bbbk} by $A_k^\dagger$ to obtain
\be
\I \otimes \widetilde{\mb{B}}_k \ket{\psi}_\mc{AB}=A_k^\dagger\otimes\I_d\ket{\psi}_\mc{AB}.
\ee
We can multiply the above equation by $(\gamma_{y'k}^{(1)})^*\coloneqq \gamma_{y'k}^*$ and use definition \eqref{Bk_bold} to write
\be
\left(\I \otimes \sum_{y=0}^{d-1}\gamma_{yk}\gamma_{y'k}^*\mb{B}_y\right) \ket{\psi}_\mc{AB}=\gamma_{y'k}^*A_k^\dagger\otimes\I_d\ket{\psi}_\mc{AB}.
\ee
By summing over $k$ on both sides and from the fact that
\be
\sum_{k=0}^{d-1} \gamma_{yk}^{(n)}\left(\gamma_{y'k}^{(n)}\right)^*=\delta_{yy'},
\ee
we have
\be
\I\otimes\mb{B}_y\ket{\psi}_\mc{AB}=\sum_{k=0}^{d-1}\gamma_{yk}^*A_k^\dagger\otimes\I_d\ket{\psi}_\mc{AB},
\ee
where we replaced $y$ by $y'$ for easy of reading. Now, let us multiply the above equation by its own Hermitian conjugate to obtain
\be
\I\otimes\mb{B}_y\mb{B}_y^\dagger\ket{\psi}_\mc{AB}=\sum_{k,k'=0}^{d-1}\gamma_{yk}^*\gamma_{yk'}A_k^\dagger A_{k'}\otimes\I_d\ket{\psi}_\mc{AB}.
\ee
On the r.h.s. of the above we can explicitly calculate the terms
\be
\sum_{k,k'=0}^{d-1}\gamma_{yk}^*\gamma_{yk'}A_k^\dagger A_{k'}=\sum_{k,k'=0}^{d-1}\gamma_{yk}^*\gamma_{yk'}Z^{k-k'}=\sum_{j,k,k'=0}^{d-1}\gamma_{yk}^*\gamma_{yk'}\omega^{j(k-k')}\ket{j}\bra{j}.
\ee
To prove that the above diagonal operator is equal to the identity, we can directly calculate its coefficients
\be
\sum_{k,k'=0}^{d-1}\gamma_{yk}^*\gamma_{yk'}\omega^{j(k-k')}=\left|\sum_{k=0}^{d-1} \gamma_{yk}^*\omega^{jk} \right|^2=\left|\sum_{k,l=0}^{d-1}\frac{1}{d}e^{\mathbbm{i}\phi_{l\oplus 1,y}}\omega^{k(j-l)} \right|^2=\left| e^{\mathbbm{i}\phi_{j\oplus 1},y} \right|^2=1, \qquad \forall j,y.
\ee
Therefore $\mb{B}_y\mb{B}_y^\dagger=\I_d$. Similar arguments result in $\mb{B}_y^\dagger\mb{B}_y = \I_d$, which means that $\mb{B}_y$ must be unitary. Since $B_y$ is unitary as well, we can see from the block decomposition \eqref{block} that $C_y=D_y=0$.

\textit{(Step 2.)} Now, let us prove the self-testing statement. First, let us write $\ket{\psi}_\mc{AB}$ in its Schmidt decomposition form
\be\label{schmidt_decomposition}
\ket{\psi}_\mc{AB}=\sum_{i=0}^{d-1}\lambda_i\ket{u_i}\ket{v_i},
\ee
where $\{\ket{u_i}\}$ and $\{\ket{v_i}\}$ are orthonormal bases of $\mathbb{C}^d$ and $\mc{H}_\mc{B}$, respectively. Also, because $\rank(\rho_\mc{A})=d$, we have $\lambda_i> 0$, $\forall i$, such that $\sum_i \lambda_i^2=1$. Note that, there is a unitary transformation $U_\mc{B}$ that satisfies $U_\mc{B}\ket{v_i}=\ket{u_i^*}$, $\forall i$, where $^*$ denotes complex conjugation. Therefore, we can rewrite \eqref{schmidt_decomposition} as
\be
\ket{\psi}_\mc{AB}=\left(P_\mc{A}\otimes U_\mc{B}^\dagger \right)\dfrac{1}{\sqrt{d}}\sum_{i=0}^{d-1}\ket{u_i}\ket{u_i^*},
\ee
where $P_\mc{A}$ is an operator that is diagonal in the basis $\{u_i\}$, with eigenvalues $\sqrt{d}\lambda_i$. Note also that because the operators $A_i=XZ^{-i}$ form a set of genuinely incompatible measurements, we have that $\rank(\rho_\mc{A}) = d$ (see Supplemental Material of Ref. \cite{sarkar2022certification} for proofs). Therefore, we can write $\ket{\psi}_\mc{AB}$ in terms of a maximally entangled state $\ket{\phi_d^+}$ as
\be\label{PUphi}
\ket{\psi}_\mc{AB}=\left(P_\mc{A}\otimes U_\mc{B}^\dagger \right)\ket{\phi_d^+}.
\ee
Now, we can employ the above equation in \eqref{Ak_Bbbk} to obtain
\be\label{APBUdagger}
\left(A_k P_\mc{A} \otimes U_\mc{B}\widetilde{\mb{B}}_k U_\mc{B}^\dagger\right) \ket{\phi^+_d}_\mc{AB}=\left(P_\mc{A}\otimes \I \right)\ket{\phi_d^+}.
\ee
A quite useful mathematical fact is the following: If $R$ and $Q$ are any two $d\times d$ matrices, then $R\otimes Q\ket{\phi_d^+}=RQ^T\otimes \I\ket{\phi_d^+}$. By employing this fact to \eqref{APBUdagger} and omitting the state, we obtain
\be\label{APUBUtranspose}
A_k P_\mc{A}\left(U_\mc{B}\widetilde{\mb{B}}_k U_\mc{B}^\dagger \right)^\text{T}=P_\mc{A}.
\ee
If we combine the above equation with its Hermitian conjugate
\be
\left(U_\mc{B}\widetilde{\mb{B}}_k U_\mc{B}^\dagger \right)^* P_\mc{A} A_k^\dagger =P_\mc{A},
\ee
we get
\be
A_k P_\mc{A}^2A_k^\dagger=P_A^2,
\ee
where we have used the fact that $\widetilde{\mb{B}}_k$ is unitary and $P_\mc{A}$ is self-adjoint. The above equation is equivalent to $[A_k,P_\mc{A}^2]=0$, because every $A_k$ is unitary. Also, since $P_\mc{A}$ is positive semidefinite, this further imply that $[A_k,P_\mc{A}]=0$. By Lemma 1 of the Supplemental Material of Ref. \cite{sarkar2022certification}, we can conclude that $P_\mc{A}$ is proportional to identity. By using this fact in \eqref{PUphi}, we immediately obtain \eqref{state_self_test}.

Now, let us self-test the measurements performed by Bob. From the fact that $P_\mc{A}$ is proportional to identity, we can conclude from \eqref{APUBUtranspose} that
\be
A_k\left(U_\mc{B}\widetilde{\mb{B}}_k U_\mc{B}^\dagger \right)^\text{T}=\I.
\ee
By applying $A_k^\dagger$ on both sides of the above equation and taking the transpose, we obtain
\be
U_\mc{B} \widetilde{\mb{B}}_k U_\mc{B}^\dagger =A_k^*.
\ee
From definition \eqref{substitution_rule}, the above equation can be expressed as
\be\label{sum_byk_UBU}
\sum_{y=0}^{d-1} \gamma_{yk}^{(1)} U_\mc{B} \mb{B}_y U_\mc{B}^\dagger =A_k^*.
\ee
We can now multiply both sides by $\left(\gamma_{y'k}^{(1)}\right)^*$, sum over $k$, and use Eq. \eqref{b_delta} to obtain
\be
\sum_{y=0}^{d-1} \delta_{yy'} U_\mc{B}\mb{B}_y U_\mc{B}^\dagger =\sum_{k=0}^{d-1} \left(\gamma_{y'k}^{(1)} \right)^* XZ^k.
\ee
Finally, from definition \eqref{Bbar}, the above equation results in
\be
U_\mc{B}\mb{B}_y U_\mc{B}^\dagger = \overline{B}_y,
\ee
which finishes the proof.
\end{proof}

Now, we are going to extend the above results to the case where we do not assume that (i) the measurements performed by Bob are projective and (ii) that the state shared by the parties is pure.

\begin{theorem}\label{Theorem_2}
    Assume that the Steering inequality \eqref{Steering_inequality} is maximally violated by a state $\rho_\mc{AB}$ acting on $\mathbb{C}^d\otimes\mc{H_B}$ and unitary $d$-outcome observables $A_k$ and $B_y$ ($k,y\in\{0,\ldots,d-1\}$) acting on, respectively, $\mathbb{C}^d$ and $\mc{H_B}$ such that the observables on Alice's trusted side are given by $A_k=XZ^{-k}$. Then, the following three statements hold true:

    (i) Bob's measurements are projective, that is, all operators $B_{n|y}$ are unitary such that $B_{n|y}^d = \I$, 
    
    (ii) Bob's Hilbert space decomposes as $\mc{H_B}=(\mb{C}^d)_{\mc{B}'}\otimes \mc{H}_{\mc{B}''}$, and

    (iii) there exists a local unitary transformation on Bob's untrusted side, $U_\mc{B}: \mc{H_B}\to \mc{H_B}$, such that
\be\label{state_self_test2}
(\I_\mc{A}\otimes U_\mc{B})\rho_\mc{AB}(\I_\mc{A}\otimes U_\mc{B}^\dagger)=\ket{\phi_d^+}\bra{\phi_d^+}_\mc{AB'}\otimes\rho_{\mc{B}'}
\ee
and
\be\label{operator_self_test2}
\forall y,\quad U_\mc{B} B_y U_\mc{B}^\dagger=\overline{B}_y\otimes\I_\mc{B''},
\ee
where $\mc{B}''$ denotes Bob's auxiliary system and $\overline{B}_y\coloneqq \overline{B}_y^{(1)}$ are given by \eqref{Bbarn}.
\end{theorem}

\begin{proof}
    \textit{(Step 1.)} Let us start by proving that if the Steering inequality is maximally violated, then Bob's measurements must be projective. From the SOS decomposition \eqref{SOS_decomposition}, if the Steering inequality is maximally violated by a set of operators $B_{n|y}$ and a state $\rho_\mc{AB}$, then
    \be\label{AkBkrho}
    \left(A_k^n\otimes \widetilde{B}_{k}^{(n)}\right)\rho_\mc{AB}=\rho_\mc{AB},
    \ee
    for $n=1,\ldots,d-1$. In particular,
    \be
    \left(A_k^{d-n}\otimes \widetilde{B}_{k}^{(d-n)}\right)\rho_\mc{AB}=\rho_\mc{AB}
    \ee
    also holds. From \eqref{Bn^dagger=Bd-n} and Fact \ref{fact_b}, we can see that
    \be
    \widetilde{B}_k^{(d-n)}=\left(\widetilde{B}_k^{(n)} \right)^\dagger
    \ee
    is also valid even if we do not assume that the measurements are projective. The three equations above give us
    \be
    \left[\I_\mc{A}\otimes \widetilde{B}_k^{(n)}\left(\widetilde{B}_k^{(n)}\right)^\dagger\right]\rho_\mc{AB}=\rho_\mc{AB},
    \ee
    which immediately implies that
    \be\label{Btilde_unitary}
    \widetilde{B}_k^{(n)}\left(\widetilde{B}_k^{(n)}\right)^\dagger=\I_\mc{B}.
    \ee
    Similar arguments lead to $\left(\widetilde{B}_k^{(n)}\right)^\dagger\widetilde{B}_k^{(n)}=\I_\mc{B}$, so we can conclude that every $\widetilde{B}_k^{(n)}$ is a unitary operator. From definition \eqref{substitution_rule}, the above equation becomes
    \be
    \sum_{y,y'=0}^{d-1}\gamma_{yk}^{(n)}\left(\gamma_{y'k}^{(n)} \right)^* B_{n|y} B_{n|y'}^\dagger = \I_\mc{B}.
    \ee
    If we some over $k$ on both sides of the above equation, we can use \eqref{b_delta} to obtain
    \be\label{BB=d1}
    \sum_{y=0}^{d-1} B_{n|y} B_{n|y}^\dagger = d\I_\mc{B}.
    \ee
    Because $B_{n|y}B_{n|y}^\dagger\leqslant \I_\mc{B}$ for any $n$ and any $y$, then \eqref{BB=d1} implies that
    \be
    B_{n|y} B_{n|y}^\dagger = \I_\mc{B}.
    \ee
    Similar arguments can be used to see that $B_{n|y}^\dagger B_{n|y} = \I_\mc{B}$, which allows us to conclude that every operator $B_{n|y}$ is unitary. In that case, Bob's measurements $\{N_{b|y}\}$ must be projective. As in the previous theorem, we can now use the notation $B_{n|y}=B_y^n$.
    
    \textit{(Step 2.)} To prove the main result, let us suppose that the state $\rho_\mc{AB}$ that maximally violates the Steering inequality \eqref{Steering_inequality} admits the decomposition
    \be\label{decomposition_psi_s}
    \rho_\mc{AB}=\sum_{s=1}^K p_s\ket{\psi_s}\bra{\psi_s}_\mc{AB},
    \ee
    where $p_s\geqslant 0$ for every $s$ such that $\sum_s p_s=1$ and $K$ is a positive integer. Without loss of generality, we can assume that $\ket{\psi_s}$ are the eigenvectors of $\rho_\mc{AB}$ and, therefore, $\braket{\psi_i}{\psi_j}=\delta_{ij}$. Because $\rho_\mc{AB}$ maximally violates \eqref{Steering_inequality}, every state $\ket{\psi_s}_\mc{AB}$ must maximally violate it too. Therefore,
    \be\label{AkBk_psi_s}
    \left(A_k^n\otimes \widetilde{B}_k^{(n)} \right)\ket{\psi_s}_\mc{AB}=\ket{\psi_s}_\mc{AB},
    \ee
    for every $k=0,\ldots,d-1$ and $n=1,\ldots,d-1$. From Theorem \ref{theorem_1}, there must be a unitary $U_\mc{B}^{(s)}$ such that
    \be\label{Upsi_s=phi}
    \left(\I_\mc{A}\otimes U_\mc{B}^{(s)}\right)\ket{\psi_s}_\mc{AB}=\ket{\phi_d^+},
    \ee
    for every $s$. Note that the unitaries $U_\mc{B}^{(s)}$ are not necessarily equal for different $s$.

    Let us now proceed similarly to the previous proof and write the state $\ket{\psi_s}_\mc{AB}$ in terms of the maximally entangled state $\ket{\phi_d^+}$ in the following way
    \be
    \ket{\psi_s}_\mc{AB}=\I_\mc{A}\otimes \left(U_\mc{B}^{(s)}\right)^\dagger\dfrac{1}{\sqrt{d}}\sum_i\ket{u_i}\ket{u_i^*}=\dfrac{1}{\sqrt{d}}\sum_i \ket{u_i}\ket{f_i^s},
    \ee
    where the states $\ket{f_i^s}=\left(U_\mc{B}^{(s)}\right)^\dagger\ket{u_i^*}$ form an orthonormal basis for every $s$. In addition, we can also conclude from Theorem \ref{theorem_1} that
    \be
    U_\mc{B}^{(s)}\mathbb{B}_y^{(s)} \left( U_\mc{B}^{(s)}\right)^\dagger=\overline{B}_y,
    \ee
    for every $y$, where $\mathbb{B}_y^{(s)}\coloneqq\Pi_\mc{B}^s B_y \Pi_\mc{B}^s$ is the operator $B_y$ projected onto the support of Bob's local state $\rho_\mc{B}^{(s)}=\Tr_\mc{A}\left(\ket{\psi_s}\bra{\psi_s}_\mc{AB} \right)$. This support is characterized as
    \be
    \text{supp}\left(\rho_\mc{B}^{(s)} \right)\equiv V_s=\text{span}\left\{\ket{f_0^s},\ldots,\ket{f_{d-1}^s} \right\}\subset \mc{H_B}.
    \ee
    Similarly to the proof of Theorem \ref{theorem_1}, we can now affirm that Bob's observables $B_y$ can be decomposed as
    \be
    B_y=\mathbb{B}_y^{(s)}\oplus E_y^{(s)},
    \ee
    where $\mathbb{B}_y^{(s)}$ are unitary operators such that $\left(\mathbb{B}_y^{(s)}\right)^d=\I_d$ and $E_y^{(s)}$ are also unitaries acting on the complement of $V_s$. Now, we can multiply the above equation by $\gamma_{yk}^{(1)}$ on both sides and take the summation over $y$ to obtain
    \be\label{Bk_direct_sum}
    \widetilde{B}_k=\widetilde{\mathbb{B}}_k^{(s)}\oplus \widetilde{E}_k^{(s)},
    \ee
    where
    \be
    \widetilde{\mathbb{B}}_k^{(s)}\coloneqq \sum_{y=0}^{d-1}\gamma_{yk}^{(1)}\mathbb{B}_y^{(s)},
    \ee
    and similarly to $\widetilde{E}_k^{(s)}$. Note that we already proved that $\widetilde{B}_k=\widetilde{B}_k^{(1)}$ is unitary in \eqref{Btilde_unitary}. In addition, we can use the same arguments from Theorem \ref{theorem_1} in \eqref{AkBk_psi_s} to show that $\widetilde{\mathbb{B}}_k^{(s)}$ is also unitary. From \eqref{Bk_direct_sum}, we can follow the same steps in the proof of Theorem 1.2 of Ref. \cite{sarkar2022certification} to conclude that the local supports $V_s$ are mutually orthogonal. That implies that Bob's local Hilbert space admits the following decomposition
    \be\label{orthogonal_subspaces}
    \mc{H_B}=V_1\oplus V_2\oplus \ldots\oplus V_{k}=(\mathbb{C}^d)_{\mc{B}'}\otimes \mc{H}_{\mc{B}''},
    \ee
    where the second equality is due to the fact that $\dim V_s=d$ for every $s=1,\ldots,K$.

    \textit{(Step 3.)} Now, let us prove equations \eqref{state_self_test2} and \eqref{operator_self_test2} by following the same procedure done in the proof of Theorem 1.2 of Ref. \cite{sarkar2022certification}. For every $s$, the states $\{\ket{f_k^s}\}_i$ span orthogonal subspaces. For this reason, we can affirm that there exists a single unitary $U_\mc{B}:\mc{H_B}\to\mc{H_B}$ such that
    \be
    U_\mc{B}\ket{f_i^s}=\ket{i}_\mc{B'}\otimes\ket{s}_\mc{B''},
    \ee
    for $i=0,\ldots,d-1$ and $s=0,\ldots,K$. Therefore, we conclude that
    \be
    (\I_\mc{A}\otimes U_\mc{B})\ket{\psi_s}_\mc{AB}=\ket{\phi_s^+}_\mc{AB'}\otimes \ket{s}_\mc{B''}.
    \ee
    Now, we can apply $U_\mc{B}$ on \eqref{decomposition_psi_s} to obtain
    \be\label{selftest_mixed_state}
    (\I_\mc{A}\otimes U_\mc{B})\rho_\mc{AB}(\I_\mc{A}\otimes U_\mc{B}^\dagger)=\ket{\phi_d^+}\bra{\phi_d^+}_\mc{AB'}\otimes \rho_\mc{B''},
    \ee
    where $\rho_\mc{B''}=\sum_s p_s\ket{s}\bra{s}_\mc{B''}$, as we wanted to prove.

    From \eqref{orthogonal_subspaces}, we can conclude that the operators $\widetilde{B}_k$ admit the following decomposition under the action of $U_\mc{B}$
    \be\label{U_B_decomposed}
    U_\mc{B} \widetilde{B}_k U_\mc{B}^\dagger= \sum_{i,j}\widetilde{B}_{i,j,k}\otimes \ket{i}\bra{j}_\mc{B''}, 
    \ee
    where $\widetilde{B}_{i,j,k}$ are matrices of dimension $d$ acting over $(\mathbb{C}^d)_\mc{B'}$ and $\ket{s}_\mc{B''}$ is the computational basis of $\mc{H}_\mc{B''}$ and, at the same time, the eigenbasis of $\rho_\mc{B''}$. From \eqref{AkBkrho} for $n=1$, we can write that
    \be
    \left[A_k\otimes \left(U_\mc{B}\widetilde{B}_k U_\mc{B}^\dagger\right)\right] U_\mc{B}\rho_\mc{AB} U_\mc{B}^\dagger=U_\mc{B}\rho_\mc{AB} U_\mc{B}^\dagger.
    \ee
    By applying \eqref{selftest_mixed_state} and \eqref{U_B_decomposed} in the above, we obtain
    \be\label{ABpsi_pij}
    \sum_{i,j} \left(A_k \otimes \widetilde{B}_{i,j,k}\right)\left(\ket{\phi_d^+}\bra{\phi_d^+}_\mc{AB'}\otimes p_j\ket{i}\bra{j}\right)=\ket{\phi_d^+}\bra{\phi_d^+}_\mc{AB'}\otimes \rho_\mc{B''}.
    \ee
    The above matrix equation can only be satisfied if $\left(A_k \otimes \widetilde{B}_{i,j,k}\right)\ket{\phi_d^+}=0$ for $i\neq j$. Since $A_k$ is invertible, we must have $\widetilde{B}_{i,j,k}=0$ for $i\neq j$. Thus, the remaining diagonal terms of \eqref{ABpsi_pij} must satisfy
    \be
    \left(A_k\otimes\widetilde{B}_{i,i,k}\right)\ket{\phi_d^+}=\ket{\phi_d^+}.
    \ee
    We know that given matrices $Q$ and $R$, we always have $Q\otimes R\ket{\psi_d^+}=QR^T\otimes \I\ket{\phi_d^+}$. Therefore, we can conclude from the above equation that $A_k\left(\widetilde{B}_{i,i,k}\right)^T=\I_d$, which immediately implies that $\widetilde{B}_{i,i,k}=A_k^*$. Having said that, equation \eqref{U_B_decomposed} becomes
    \be
    U_\mc{B} \widetilde{B}_k U_\mc{B}^\dagger= A_k^*\otimes \sum_i\ket{i}\bra{i}_\mc{B''}=A_k^*\otimes \I_\mc{B''}.
    \ee
    Finally, we can apply the same procedure that is done in \eqref{sum_byk_UBU} to obtain
    \be
    U_\mc{B} B_y U_\mc{B}^\dagger=\overline{B}_y\otimes \I_\mc{B''},
    \ee
    as we wanted to prove.
\end{proof}

\section{Solutions for the qutrit case}\label{appendix_Uniqueness}

Here we prove that in the three-dimensional case ($d=3$) the solutions 
to the systems of equations \eqref{b_delta_main} can be found analytically 
and is of the form \eqref{conditions}. 

The equations \eqref{b_delta_main} can be thought of as the orthogonality 
conditions for three three-dimensional vectors, each composed of phases. 
Let us restate them as
\begin{equation}
    \sum_{i=0}^2\mathrm{exp}[{\mathrm{i}(\phi_{i,0}-\phi_{i,1})}]=0,\qquad
    \sum_{i=0}^2\mathrm{exp}[{\mathrm{i}(\phi_{i,2}-\phi_{i,0})}]=0,\qquad
    \sum_{i=0}^2\mathrm{exp}[{\mathrm{i}(\phi_{i,1}-\phi_{i,2})}]=0.
\end{equation}
Denoting then $p_i=\phi_{i,0}-\phi_{i,1}$, $q_i=\phi_{i,0}-\phi_{i,2}$ and
$r_i=\phi_{i,1}-\phi_{i,2}$, we can further simplify these equations to 
\begin{equation}\label{equations2}
    \sum_{i=0}^2\mathrm{exp}(\mathrm{i}p_i)=\sum_{i=0}^2\mathrm{exp}(\mathrm{i}q_i)=\sum_{i=0}^2\mathrm{exp}(\mathrm{i}r_i)=0,
\end{equation}
where, the new variables $p_i$, $r_i$ and $q_i$ satisfy additionally 
\begin{equation}\label{addconditions}
    p_i+q_i+r_i=0\qquad (i=0,1,2).
\end{equation}

Let us now find solutions to the first of these equations and rewrite it as
\begin{equation}\label{equation3}
  \mathrm{exp}(\mathrm{i}p_0)+\mathrm{exp}(\mathrm{i}p_1)  =-\mathrm{exp}(\mathrm{i}p_3).
\end{equation}
Let us then multiply this equation by its conjugation to obtain
\begin{equation}
  \left|\mathrm{exp}(\mathrm{i}p_0)+\mathrm{exp}(\mathrm{i}p_1)\right|^2  =1.
\end{equation}
This can further be rewritten as $\cos(p_0-p_1)=-1/2$ which has countably many solutions, which under exponentiation are equivalent to are $p_0=p_1\pm 2\pi/3$. Let us additionally notice that since the angles $\phi_{i,y}$ must satisfy $\sum_{i}\phi_{i,y}=0$ for any $y$, then we must have
\begin{equation}
    \sum_{i=0}^2p_i=0.
\end{equation}
It is not difficult to see that the above implies that the possible solutions for $p_0$
are $p_1=0,\pm 2\pi/3$. Let us then notice that due to the fact that the equations 
\eqref{equations2} can always be multiplied by $\omega$ or $\omega^2$, we can fix $p_0$
to be, say $p_0=2\pi/3$. Taking all this into account, we obtain two possible solutions of Eq. (\ref{equation3}) which, up to exponentiations, can be stated as
\begin{equation}
    p_0=2\pi/3,\qquad p_1=0,\qquad p_2=-2\pi/3 \qquad\mathrm{or}\qquad p_0=2\pi/3,\qquad p_1=-2\pi/3,\qquad p_2=0.
\end{equation}

The same arguments can be applied to the second equation in \eqref{equations2}, however, 
here we choose $q_0=-2\pi/3$, which gives
\begin{equation}
    q_0=-2\pi/3,\qquad q_1=2\pi/3,\qquad q_2=0 \qquad\mathrm{or}\qquad q_0=-2\pi/3,\qquad q_1=0,\qquad q_2=2\pi/3.
\end{equation}
To finally fix $r_i$ we must take into account Eq. (\ref{addconditions}), but also the fact that 
$r_0+r_1+r_2=0$. All this implies that the possible solutions are 
\begin{equation}\label{first}
    p_0=2\pi/3,\quad p_1=0,\quad p_2=-2\pi/3,\qquad q_0=-2\pi/3,\quad q_1=2\pi/3,\quad q_2=0,\qquad r_0=0,\quad r_1=-2\pi/3,\quad r_2=2\pi/3.
\end{equation}
or
\begin{equation}\label{second}
  p_0=2\pi/3,\quad p_1=-2\pi/3,\quad p_2=0,\qquad  q_0=-2\pi/3,\quad q_1=0,\quad q_2=2\pi/3,\qquad r_0=0,\quad r_1=2\pi/3,\quad r_2=-2\pi/3.
\end{equation}
On the level of the angles $\phi_{i,y}$, the first of these two solutions corresponds to 
Eq. (\ref{conditions}) and gives rise to the observables presented in Eqs. (\ref{1}), (\ref{2}) and (\ref{3}), whereas the second one to 
\be
\begin{array}{c}
y=1:\left\{\begin{array}{l}
\phi_{0,1}=\phi_{0,0}-\dfrac{2\pi}{3},\\[1ex]
\phi_{1,1}=\phi_{1,0}+\dfrac{2\pi}{3},
\end{array}\right.\\
y=2:
\left\{
\begin{array}{l}
\phi_{0,2}=\phi_{0,0}-\dfrac{2\pi}{3},\\[1ex]
\phi_{1,2}=\phi_{1,0},
\end{array}\right.
\end{array}
\label{conditions23}
\ee
which gives rise to slightly modified observables of the following form
\begin{subequations}
\begin{align}
\overline{B}_1 &\equiv\begin{pmatrix}
 0 & 0 & \omega^2 e^{\mathbbm{i} \phi _{0,0}} \\
 \omega e^{\mathbbm{i} \phi _{1,0}} & 0 & 0 \\
 0 & e^{-\mathbbm{i} \left(\phi _{0,0}+\phi _{1,0}\right)} & 0 
\end{pmatrix},\\
\overline{B}_2 &\equiv\begin{pmatrix}
 0 & 0 & \omega^2 e^{\mathbbm{i} \phi _{0,0}} \\
  e^{\mathbbm{i} \phi _{1,0}} & 0 & 0 \\
 0 & \omega e^{-\mathbbm{i} \left(\phi _{0,0}+\phi _{1,0}\right)} & 0 
\end{pmatrix}
\end{align}
\end{subequations}
with $\overline{B}_0$ being the same as in the previous case.

%
%

\section{Robust self-testing}\label{robustness}

Let us now prove that when the maximal violation of the inequality is not attained by an $\epsilon$, then the measurements performed by Bob, i.e. $B_y$, are close to the reference measurements $\overline{B}_y$ by a factor as a function of $\epsilon$. For simplicity, we will assume that the state is pure and the measurements performed by Bob are projective. 

To prove the robustness of our self-testing statement for the qutrit case, we will make use of the fact that
\begin{equation}\label{rownanie}
    \widetilde{B}_k^{(n)}=\widetilde{B}_k^n, \quad \forall n\in\{1,2\}\quad\text{and}\quad d=3,
\end{equation}
which is true from the fact that $\widetilde{B}_k^{(2)}=\left(\widetilde{B}_k^{(1)}\right)^\dagger$ and $\left(\widetilde{B}_k^{(2)} \right)^3=\I$. We can conjecture that the above relation is also valid for any prime $d>3$:
\begin{conjecture}\label{conjecture}
    For any prime $d$, we have that $\widetilde{B}_k^{(n)}=\widetilde{B}_k^n$, for every $n\in\{1,\ldots,d-1\}$. Note that this statement is true for $d=3$. 
\end{conjecture}

Having said that, we proceed to the robustness analysis by assuming the above conjecture, keeping in mind that (\ref{rownanie}) holds true for
$d=3$.

\begin{theorem}
Consider the unitary $d$-outcome observables $A_k=XZ^{-k}$ and $B_y$, for $k,y\in{0,\ldots,d-1}$ acting on, respectively $\mathbb{C}^d$ and $\mc{H}_\mc{B}$. If the Steering inequality \eqref{Steering_inequality}
\be
\langle\mc{S}\rangle=\sum_{n=1}^{d-1}\sum_{k=0}^{d-1} \langle A_k^n\otimes\widetilde{B}_k^{(n)} \rangle \leqslant \beta_C,
\ee
is violated by a state $\ket{\psi}_\mc{AB} \in \mathbb{C}^{d} \otimes \mathcal{H}_{B} $ and observables $ B_{y} $ such that $\langle\mc{S}\rangle \geqslant d(d-1) - \epsilon > \beta_C$, then there exists a unitary operation $ U_\mc{B}: \mathcal{H}_\mc{B} \to \mathcal{H}_\mc{B} $ such that
\begin{equation}\label{distanceUBpsi}
    \left\|\left(\I_\mc{A}\otimes U_\mc{B}\right)\left( \I_\mc{A}\otimes B_y \right)\ket{\psi}-\I\otimes \overline{B}_y\ket{\phi_d^+}\right\|\leqslant 2\sqrt{d}(2\epsilon)^\frac{1}{4}
\end{equation}
and
\begin{equation}\label{distanceBBbar}
    \left\|B_y-\overline{B}_y\right\|_2\leqslant 4d\left(2\epsilon \right)^\frac{1}{4},
\end{equation}
where $ k = 0, \ldots, d-1 $ and $ \overline{B}_{y}$ are Bob's ideal observables given by \eqref{Bbar}, and $ \| \cdot \|_{2} $ stands for the Hilbert-Schmidt norm.
\end{theorem}
\begin{proof}
    This proof follows similar steps to Ref. \cite{sarkar2022certification}. We write them all here for completion, adapted to the operators \eqref{Bbar}. From the violation of the Steering inequality, we have that
\begin{equation}\label{Re(psi)}
    \left|\bra{\psi} A_k^n\otimes \widetilde{B}_k^{(n)}\ket{\psi} \right| \geqslant \text{Re} \left( \bra{\psi} A_k^n\otimes \widetilde{B}_k^{(n)}\ket{\psi} \right) \geqslant 1-\epsilon,
\end{equation}
for $k=0,\ldots,d-1$ and $n=1,\ldots,d-1$. Now, let us consider Alice's observables $A_k$ and rewrite them in terms of their own eigenvectors $\ket{q_j}$
\be
A_k=\sum_j \omega^j \ket{q_j^{(k)}}\bra{q_j^{(k)}},
\ee
where
\be
\ket{q_j^{(k)}}=\sum_{m=0}^{d-1} \omega^{-km(m-1)/2-jm}\ket{m}.
\ee
For a given $k$, we can write the operators $X$, $Z^{-l}$ and their product in the basis $\{\ket{q_j^{(k)}}\}_j$ as
\begin{equation}
    X=\sum_{j=0}^{d-1} \omega^j\ket{q_j^{(k)}}\bra{q_{j+k}^{(k)}},\qquad Z^{-l}=\sum_{j=0}^{d-1} \ket{q_{j+l}^{(k)}}\bra{q_j^{(k)}}, \qquad XZ^{-l}=\sum_{j=0}^{d-1} \omega^{j+l-k}\ket{q_{j+l-k}^{(k)}}\bra{q_j^{(k)}}.
\end{equation}
Now, let us consider a particular basis $\{|q_j^{(k)}\rangle\}_{j}$ for a fixed $k$, and let us rewrite the state $\ket{\psi}$ in the basis $\{\ket{q_j^{(k)}}\ket{b_j}\}_j$ as
\begin{equation}\label{psi_qb}
    \ket{\psi}=\sum_{i=0}^{d-1}\alpha_i \ket{q_i^{(k)}}\ket{b_i},
\end{equation}
where $\alpha_i$ are nonnegative real numbers that satisfy $\sum_{i=0}^{d-1}\alpha^i=1$ and $\ket{b_i}$ are not necessarily orthogonal vectors in space $\mc{H_B}$. Now, observe that
\begin{equation}\label{XZ-l}
    \left(XZ^{-l} \right)^n \ket{q_j^{(k)}}=\omega^{nj+n(n+1)(l-k)/2} \ket{q_{j+n(l-k)}^{(k)}}
\end{equation}
for $k,l=0,\ldots,d-1$ and $n=1,\ldots,d-1$. By using Eqs. \eqref{psi_qb} 
and \eqref{XZ-l} in Eq. \eqref{Re(psi)} we have   
\begin{equation}\label{Re(Btilde(n))}
    \sum_{i=0}^{d-1}\alpha_{i+n(l-k)}\alpha_i\text{Re}\left( \omega^{ni+n(n+1)(l-k)/2} \bra{b_{i+n(l-k)}} \widetilde{B}_k^{(n)}\ket{b_i} \right)\geqslant 1-\epsilon.
\end{equation}
From Conjecture \ref{conjecture}, we have that
\begin{equation}
    \widetilde{B}_k^{(n)}=\widetilde{B}_k^n, \quad \forall n\in\{1,\ldots,d-1 \}.
\end{equation}
From the above, Eq. \eqref{Re(Btilde(n))} becomes
\begin{equation}\label{Re(Btilde^n)}
    \sum_{i=0}^{d-1}\alpha_{i+n(l-k)}\alpha_i\text{Re}\left( \omega^{ni+n(n+1)(l-k)/2} \bra{b_{i+n(l-k)}} \widetilde{B}_k^n\ket{b_i} \right)\geqslant 1-\epsilon.
\end{equation}

Let us now consider some specific cases of the above inequality to get some bounds regarding the coefficients $\alpha_i$ of the state \eqref{psi_qb}.
If $k=l+1$, we obtain
\begin{equation}\label{alpha_i+n}
    \sum_{i=0}^{d-1}\alpha_{i+n}\alpha_i\text{Re}\left( \omega^{ni+n(n+1)/2} \bra{b_{i+n}} \widetilde{B}_k^n\ket{b_i} \right)\geqslant 1-\epsilon.
\end{equation}
Because $\widetilde{B}_k$ is unitary and $\ket{b_i}$ are normalized, we have that
\be
\left| \omega^{ni+n(n+1)/2} \bra{b_{i+n}} \widetilde{B}_k^n\ket{b_i} \right| \leqslant 1.
\ee
Since every complex number $z$ satisfies $\text{Re}(z)\leqslant |z|$, the above equation and \eqref{alpha_i+n} imply that
\be
\sum_{i=0}^{d-1} \alpha_{i+n}\alpha_i \geqslant 1-\epsilon, \qquad \forall n\in\{0,\ldots,d-1\}.
\ee
Note that if we sum over $n$ in the above inequality, we obtain
\be\label{sum_alpha_i}
\sum_i \alpha_i \geqslant \sqrt{d}\sqrt{1-\epsilon}.
\ee
Some extra relations can be derived from the above inequality, such as \cite{sarkar2022certification}
\be\label{<alpha_i}
\frac{1}{\sqrt{d}} - \sqrt{2\epsilon} \leq \alpha_{i} \leq \frac{1}{\sqrt{d}} + \sqrt{2\epsilon}
\ee
and
\be
\frac{1}{d} - \sqrt{2\epsilon} \leq \alpha_{i} \alpha_{i + j} \leq \frac{1}{d} + \sqrt{2\epsilon}.
\ee

Now, let us show that Bob's observables are going to be close to the reference observables. For that, let us consider Eq. \eqref{Re(Btilde^n)} when $l=k$, that is
\begin{equation}
    \sum_{i=0}^{d-1}\alpha_i^2\text{Re}\left( \omega^{ni} \bra{b_{i}} \widetilde{B}_k^n\ket{b_i} \right)\geqslant 1-\epsilon.
\end{equation}
Since $\text{Re}\left( \omega^{ni} \bra{b_{i}} \widetilde{B}_k^n\ket{b_i} \right)\leqslant 1$, we can separate one arbitrary term from the sum in the above inequality as the following
\begin{equation}
    \alpha_j^2\text{Re}\left( \omega^{nj} \bra{b_{j}} \widetilde{B}_k^n\ket{b_j} \right)+\sum_{i=0,\ i\neq j}^{d-1}\alpha_i^2\geqslant 1-\epsilon,
\end{equation}
for any $j$. Since $\alpha_j^2+\sum_{i=0,\ i\neq j}^{d-1}\alpha_i^2=1$, the above inequality becomes
\be
    \alpha_j^2\left[\text{Re}\left( \omega^{nj} \bra{b_{j}} \widetilde{B}_k^n\ket{b_j} \right)-1\right]+1\geqslant 1-\epsilon
\ee
By rearranging the terms, we obtain
\be\label{alpha_j[]}
    \alpha_j^2\left[1-\text{Re}\left( \omega^{nj} \bra{b_{j}} \widetilde{B}_k^n\ket{b_j} \right)\right]\leqslant\epsilon,
\ee
for $k,j=0,\ldots,d-1$ and $n=1,\ldots,d-1$. Now, we can use the left inequality of \eqref{<alpha_i} in \eqref{alpha_j[]} to obtain
\be\label{Re()>1-2dsqrt2e}
\text{Re}\left( \omega^{ni} \bra{b_{i}} \widetilde{B}_k^n\ket{b_i} \right)\geqslant 1-2d\sqrt{2\epsilon}.
\ee
Let us now use the eigendecomposition of $\widetilde{B}_k$ with outcomes $\omega^j$ such that
\be\label{eigendecomposition}
\widetilde{B}_k=\sum_{j=0}^{d-1}\omega^j P_{j|k},
\ee
where $P_{j|k}$ are orthogonal projectors. Now we can insert the above equation into \eqref{Re()>1-2dsqrt2e} to obtain
\begin{equation}
\sum_{j=0}^{d-1} \text{Re} \left( \omega^{n(j+i)} \langle b_{i} | P_{j|k} |b_{i} \rangle \right) \geqslant 1 - 2d\sqrt{2\epsilon}.
\end{equation}
Since $\sum_{n=0}^{d-1}\omega^{n(i+j)}=d\delta_{j,-i}$, we can sum over $n$ on both sides to get
\begin{equation}\label{<bPb>}
\bra{b_{j}} P_{-j|k} \ket{b_{j}} \geqslant 1 - 2d\sqrt{2\epsilon}.
\end{equation}
The above inequality means that, for a fixed $k$,  each vector $\ket{b_{j}}$ is close to the subspace corresponding to outcome $\omega^{-j}$ of $\widetilde{B}_k$ by a small amount $2d\sqrt{2\epsilon}$. For convenience, we are going to use the notation $\ket{u_{j}} = P_{-j|k} \ket{b_{j}}$ and $\ket{\overline{u}_{j}} = \ket{u_{j}} / ||\,\ket{u_{j}}\,||$. Because all $P_{j|k}$ are pairwise orthogonal projectors, then $\ket{\overline{u}_{j}}$ are also orthogonal. In addition, from \eqref{<bPb>} we have that $ ||\, \ket{u_j}\, || \geqslant 1 - 2d\sqrt{2\epsilon} $.

From the above discussion, we can rewrite \eqref{eigendecomposition} as
\begin{equation}\label{Bk=uu}
\widetilde{B}_k = \sum_{j=0}^{d-1} \omega^{-j} \ket{\overline{u}_{j}}\bra{\overline{u}_j} \oplus \widetilde{B}_{k}',
\end{equation}
where $\widetilde{B}_{k}'$ are operators whose support is orthogonal to every $\ket{u_j}$. Since $\ket{\overline{u}_j}$ are pair-wise orthogonal, there exists a unitary transformation $U_\mc{B}$ such that $U_\mc{B} \ket{\overline{u}_j} = \ket{q_{j}^{(k)}}$. By applying $U_\mc{B}\left(\cdot \right) U_\mc{B}^\dagger$ on both sides of \eqref{Bk=uu} we obtain
\begin{equation}
U_\mc{B} \widetilde{B}_{k} U_\mc{B}^{\dagger} = \sum_{j=0}^{d-1} \omega^{-j} \ket{q_j^{(k)}}\bra{q_j^{(k)}}\oplus \widetilde{B}_{k}''=A_k^*\oplus \widetilde{B}_k''.
\end{equation}
Similarly to the previous theorem, the above equation implies that
\be
U_\mc{B} B_y U_\mc{B}^{\dagger} = \overline{B}_y \oplus \overline{B}_y''.
\ee
We still need another ingredient before starting to work with the norms. Let us denote $\ket{b_j'} = U_\mc{B} \ket{b_{j}}$ and observe that \eqref{<bPb>} can be written as
\begin{equation}\label{bq_bound}
\braket{b_j'}{q_j^{(k)}}\geqslant 1 - 2d\sqrt{2\epsilon}.
\end{equation}

Now, let us prove \eqref{distanceUBpsi} by developing the left side of it as follows:
\be\label{bound_distance}
\left\|U_\mc{B} B_y \ket{\psi}-\overline{B}_y\ket{\phi_d^+}\right\|= \left\|U_\mc{B} B_y U_\mc{B}^\dagger U_\mc{B} \ket{\psi}-\overline{B}_y\ket{\phi_d^+}\right\| = \left\|\overline{B}_y U_\mc{B} \ket{\psi}-\overline{B}_y\ket{\phi_d^+}\right\| = \left\|U_\mc{B} \ket{\psi}-\ket{\phi_d^+}\right\|,
\ee
where the last equality is due to the unitary invariance of the norm. Note that we have omitted the identities in Alice's space. The Euclidean distance above can be calculated as follows:
\be\label{UBpsi-phi}
\left\|U_\mc{B} \ket{\psi}-\ket{\phi_d^+}\right\|=\sqrt{2\left[1-\text{Re}\left(\bra{\phi^+_d}U_\mc{B}\ket{\psi} \right) \right]}.
\ee
To calculate the bracket above, we can write the maximally entangled state of two qudits in the basis $\{\ket{q_j^{(k)}}\}_j$ as $\ket{\phi_d^+}=\frac{1}{\sqrt{d}}\sum_j \ket{q_j^{(k)}}\ket{q_j^{(k)}}$. Therefore, we have
\be
\bra{\phi^+_d}U_\mc{B}\ket{\psi}=\sum_{j=0}^{d-1}\bra{\phi_d^+}\alpha_j\ket{q_j^{(k)}} U_\mc{B}\ket{b_j}= \frac{1}{\sqrt{d}}\sum_{i,j=0}^{d-1}\alpha_j\braket{q_i^{(k)}}{q_j^{(k)}} \bra{q_i^{(k)}}U_\mc{B}\ket{b_j}= \frac{1}{\sqrt{d}}\sum_{j=0}^{d-1}\alpha_j\braket{q_j^{(k)}}{b_j'}.
\ee
If we use bound \eqref{bq_bound} in the above, we obtain
\be
\bra{\phi^+_d}U_\mc{B}\ket{\psi}\geqslant \frac{1}{\sqrt{d}}\left( 1-2d\sqrt{2\epsilon} \right)\sum_{j=0}^{d-1}\alpha_j,
\ee
which together with \eqref{sum_alpha_i} gives
\be
\bra{\phi^+_d}U_\mc{B}\ket{\psi}\geqslant \left( 1-2d\sqrt{2\epsilon} \right)\sqrt{1-\epsilon}.
\ee
By inserting the above inequality in \eqref{UBpsi-phi} we obtain
\be\label{bound_Upsi-phi}
\left\|U_\mc{B} \ket{\psi}-\ket{\phi_d^+}\right\|\leqslant \sqrt{2\left[1-\left( 1-2d\sqrt{2\epsilon} \right)\sqrt{1-\epsilon} \right]}\leqslant 2\sqrt{d}(2\epsilon)^\frac{1}{4}.
\ee
Finally, from the above and \eqref{bound_distance} we obtain 
\be
\left\|\left(\I_\mc{A}\otimes U_\mc{B}\right)\left( \I_\mc{A}\otimes B_y \right)\ket{\psi}-\I\otimes \overline{B}_y\ket{\phi_d^+}\right\|\leqslant 2\sqrt{d}(2\epsilon)^\frac{1}{4}
\ee
as desired.

In order to obtain \eqref{distanceBBbar}, first we need to observe that the Euclidean norm of an operator $A$ satisfies
\be
\left\|(\I\otimes A)\ket{\phi_d^+}\right\| = \sum_{ij} \bra{ii} \I\otimes A^\dagger A \ket{jj}=\sum_i \bra{i} A^\dagger A \ket{i}=\frac{1}{d} \Tr A^\dagger A.
\ee
Therefore, we have that
\be
\left\|B_y-\overline{B}_y\right\|_2=\sqrt{d}\left\|\left(B_y-\overline{B}_y \right)\ket{\phi_d^+} \right\|.
\ee
If we sum and subtract $B_yU_\mc{B}\ket{\psi}$ inside the norm of the right side of the above, we obtain
\be
\left\|B_y-\overline{B}_y\right\|_2=\sqrt{d}\left\|B_y\ket{\phi_d^+}-\overline{B}_y \ket{\phi_d^+}+ B_yU_\mc{B}\ket{\psi}-B_yU_\mc{B}\ket{\psi}\right\|=\sqrt{d} \left\|B_y\left(\ket{\phi_d^+}-U_\mc{B}\ket{\psi}\right)+B_yU_\mc{B}\ket{\psi}-\overline{B}_y\ket{\phi_d^+} \right\|
\ee
We can apply the triangular inequality in the last norm to have
\be
\left\|B_y-\overline{B}_y\right\|_2\leqslant \sqrt{d}\left(\left\|\ket{\phi_d^+}-U_\mc{B}\ket{\psi}\right\|+\left\|B_yU_\mc{B}\ket{\psi}-\overline{B}_y\ket{\phi_d^+}\right\| \right),
\ee
where we have used the fact that norms satisfy unitary invariance. Now, the norms on the right side of the above are bounded by terms specified by Eqs. \eqref{bound_Upsi-phi} and \eqref{distanceUBpsi}, which implies that
\be
\left\|B_y-\overline{B}_y\right\|_2\leqslant 4d\left(2\epsilon \right)^\frac{1}{4}
\ee
as we wanted to prove.
\end{proof}

\end{document}